\documentclass[comsoc]{IEEEtran}
\IEEEoverridecommandlockouts

\usepackage{cite}
\usepackage{amsmath,amssymb,amsfonts,dsfont}
\usepackage{graphicx}
\usepackage{textcomp}
\usepackage{xcolor}
\usepackage{epstopdf}
\usepackage{amsthm}
\usepackage{nicefrac}
\usepackage{balance}
\usepackage{tabularx,booktabs}
\usepackage{color}
\usepackage{subcaption}
\usepackage{mwe}
\usepackage{bbm}
\usepackage{algorithm} 
\usepackage{algpseudocode}

\newtheorem{theorem}{Theorem}

\newtheorem{remark}{Remark}

\usepackage{multicol}

\makeatletter
\if@twocolumn
\newcommand{\whencolumns}[2]{ #2 }
\else
\newcommand{\whencolumns}[2]{ #1 }
\fi
\makeatother

\algnewcommand{\algorithmicgoto}{\textbf{go to}}%
\algnewcommand{\Goto}[1]{\algorithmicgoto~\ref{#1}}%
    
\begin{document}

\title{Goodput Maximization with Quantized Feedback in the Finite Blocklength Regime for Quasi-Static Channels 
}

\author{Hasan~Basri~Celebi,~\IEEEmembership{Member,~IEEE,}
	and~Mikael~Skoglund,~\IEEEmembership{Fellow,~IEEE}
	\thanks{Hasan Basri Celebi was with the KTH Royal Institute of Technology, Stockholm, Sweden. He is now with ABB AB, Västerås, Sweden (e-mail: hasan-basri.celebi@se.abb.com). Mikael Skoglund is with the School of Electrical Engineering and Computer Science, KTH Royal Institute of Technology, Stockholm, Sweden (e-mail: skoglund@kth.se).}
	\thanks{This work was funded in part by the Swedish Foundation for Strategic Research (SSF) under grant agreement RIT15-0091.}
}

\maketitle

\begin{abstract}
	In this paper, we study a quantized feedback scheme to maximize the goodput of a finite blocklength communication scenario over a quasi-static fading channel. It is assumed that the receiver has perfect channel state information (CSI) and sends back the CSI to the transmitter over a resolution-limited error-free feedback channel. With this partial CSI, the transmitter is supposed to select the optimum transmission rate, such that it maximizes the overall goodput of the communication system. This problem has been studied for the asymptotic blocklength regime, however, no solution has so far been presented for finite blocklength. Here, we study this problem in two cases: with and without constraint on reliability. We first formulate the optimization problems and analytically solve them. Iterative algorithms that successfully exploit the system parameters for both cases are presented. It is shown that although the achievable maximum goodput decreases with shorter blocklengths and higher reliability requirements, significant improvement can be achieved even with coarsely quantized feedback schemes.  
	
\end{abstract}


\section{Introduction}\label{sec_intro}

The conventional mobile broadband communication systems are characterized to make sure that they provide high goodput and capacity to the end users. Although this is still expected in future generation communication systems, new requirements are also introduced for new applications. Three main usage scenarios are specified in the $ 5^{\text{th}} $ generation of mobile communication systems \cite{imt_requirements, celebi_wireless_comm}: (1) enhanced Mobile Broadband (eMBB), which aims at enabling communication with extremely high data rates across a wide coverage area, (2) massive Machine-Type Communication (mMTC), where the goal is to enable massive number of devices to communicate with the network infrastructure and between each other without any human interaction, and (3) Ultra-Reliable Low-Latency Communication (URLLC), which provides communication support for novel traffic types, such as mission-critical applications, with stringent constraints on reliability and end-to-end latency. 

With the recent major releases of 3GPP Radio Access Network working group, to support the URLLC applications, it is aimed at achieving reliability values not lower than \%$ 99.9999 $ with end-to-end latency not higher than $ 1 $ ms \cite{3gpp_release15,3gpp_release16,le_an_overview}. In order to meet such demands, it is required to increase the efficiency of the communication system and decrease both the latency and the overall number of packet errors. However, one major bottleneck in wireless communication is the presence of channel fading, which is caused by the multipath propagation and leads to undesired fluctuations in the received signal power, which may result in loss of the received packet. 

In this study, we consider finite blocklength communication over a frequency-flat quasi-static fading wireless channel where the channel is assumed to stay constant over some coherence time. One of the most important performance criteria over quasi-static fading channel is the overall goodput of the system, which can be measured by the achievable expected rate over a large number of packet transmissions with variable transmission rates. While studying the achievable rates over quasi-static fading channels, it is important to select the right performance limit of the system \cite{goldsmith_capacity}. For applications that are latency tolerant, i.e. no constraint on latency is introduced, it can be assumed that one packet can span over infinitely many independent fading occurences. In this case, the valid performance limit of the system can be computed with the ergodic capacity \cite{cover_elements_of}. On the other hand, for applications that are intolerant to latency, outage capacity is the valid performance limit \cite{biglieri_fading}, which is the maximal transmission rate such that the probability of the instantaneous mutual information falling below the selected transmission rate does not exceed a desired outage threshold. 

However, when stringent latency constraints are introduced, such as in URLLC systems, performance limits that are achieved in the infinite blocklength regime cannot be justified with finite blocklength \cite{ji_ultra_reliable}.  Therefore, significant amount of research on maximal achievable transmission rates for finite blocklength has been conducted in the recent years. It is shown in \cite{polyanskiy_channel_coding} that compared to the asymptotic limits, i.e. infinite blocklength regime, a rate penalty needs to be paid when transmitting in finite blocklength regime. Achievable bounds for various channels models with different fading environments have been presented in \cite{yang_quasi_static} and \cite{yang_beta_beta}.

While characterizing the achievable performance of the system over a fading channel, accurate knowledge of the channel-state-information (CSI) is required \cite{alouini_comparison}, which can be achieved by transmitting a separate training sequence that is already known to the receiver \cite{celebi_training}. However, CSI at the transmitter is also crucial for reliable communication since, with this information, resources and the transmission strategy can be adapted according to the channel state which can greatly improve the performance \cite{alouini_comparison}. In this study, it is assumed that perfect CSI is available at the receiver but only partial knowledge is available at the transmitter. This assumption is a realistic scenario for URLLC applications where CSI is crucial at the transmitter to maintain the high-reliability level however transmission of the full CSI cannot be afforded due to the latency requirements. Thus, based on this partial information, we develop an information-theoretic approach to investigate the achievable goodput of the system over a large number of blocks transmitted at variable rates.

\subsection{Related Work and Motivation}

Fixed-transmission rate communication systems with perfect CSI over slowly fading channels are studied in \cite{caire_optimum}. It is shown that the outage probability can be minimized by adapting a power control system when perfect CSI at the transmitter is available. A great amount of research has also studied the case where the transmitter has partial CSI information. Notably, the benefit of sending partial information over the feedback is to reduce the system overhead \cite{yang_noma}, compared to the case that each receiver feeds all channel information back to the transmitter. Minimization of outage probability with fixed transmission rate, when partial CSI at the transmitter is available, is studied in \cite{bhashyam_feedback}. A more systematic feedback scheme is studied in \cite{etemadi_joint} where the transmitter can only achieve quantized CSI information. In \cite{kim_on_the}, performance of a wireless communication system with partial CSI information at the transmitter, which is transmitted over a error-free quantized feedback channel, is considered in asymptotic regime. An adaptive feedback scheme is proposed that maximizes the expected rate. 

Furthermore, a one-bit feedback scheme, where the user compares the received signal strength to a threshold and sends back a binary signal informing the transmitter that if the signal power is above or below the threshold, has been studied in \cite{xu_on_the}. Partial CSI information at the transmitter for orthogonal and non-orthogonal multiple access communication systems is studied in \cite{yang_on_the} and the sufficient conditions are presented. 

On the other hand, a well-studied sequential feedback system, named as hybrid automatic repeat request, also considers a communication system with partial feedback. The performance of hybrid automatic repeat request protocols in infinite blocklength regime is studied in \cite{tuninetti_on_the}. Performance comparison between the communication systems with quantized feedback and hybrid automatic repeat request protocol is studied in \cite{makki_on_hybrid}, where a hybrid system is proposed. Recently, performance of communication systems with hybrid automatic repeat request protocol in finite blocklength regime is studied in \cite{makki_fast}. 

\subsection{Our Work and Contributions}

In this work, we propose an adaptive quantized feedback scheme in the finite blocklength regime that exploits the maximum goodput by searching the optimum selections of quantization regions with corresponding transmission rates.  We first formulate the goodput maximization problem and then analytically solve it.  An iterative algorithm that achieves the optimal feedback scheme is also proposed. A close approximation on the optimal selection of the transmission rate on each quantization region is introduced. 

Next, we study the second optimization problem being the quantized feedback scheme design to maximize the goodput with reliability constraint. In that case we first study the allocation of the maximum error allowance for each quantization region. There, we propose a waterfilling-like algorithm to assign the maximum error probability to each quantization region. We then analytically solve the optimization problem with the augmented Lagrangian method. An iterative algorithm that can achieve the maximum goodput while guaranteeing the reliability constraint is proposed. The feasibility conditions are also presented. Notably, to broaden the generality of the study, the feedback design problem is explicitly formulated, where no specific channel distribution is taken into account. Thus, results are valid for any continuous channel distribution.

\section{System Model}\label{sec_system_model}

We consider the discrete time complex baseband single-input single-output wireless communication system illustrated in Fig. \ref{fig_system_model}. A codeword of blocklength $ n $ is transmitted over a quasi-static fading channel, where the complex valued channel coefficient $ h $ is an independent and identically distributed (i.i.d.) random variable according to some distribution but remains constant over the codeword transmission. The received signal can be expressed as
\begin{equation}\label{eq_system_model}
\boldsymbol{y} = h\boldsymbol{x} + \boldsymbol{z},
\end{equation}
where $ \boldsymbol{x} $ and $ \boldsymbol{z} $ represent the $ n $-length transmitted codeword and  complex Gaussian noise vector where the samples are i.i.d and $ z_i \sim CN(0,1)$, where $z_i$ represents the $i$th component of $ \boldsymbol{z} $. Let $ \gamma = |h|^2 $ describe the i.i.d channel power. Notably, $ \gamma $ is a continuous random variable with its corresponding probability density function (pdf), denoted as $ f(\gamma) $, and cumulative distribution function (cdf), $ F(\gamma) $. It is assumed that $ f(\gamma) $ is positive over $ 0\leq\gamma\leq\infty $ and both $ f(\gamma) $ and $ F(\gamma) $ are continuous. 

It is assumed that the channel coefficient is known perfectly at the receiver. We consider that the receiver divides the positive real line into $ \Phi $ number of quantization regions and applies a deterministic index mapping on the channel power, given as  
\begin{equation}\label{key}
M(\gamma) = i, 
\end{equation}
such that $ \gamma \in [\phi_i, \phi_{i+1}) $, $ i=1,2,\cdots,\Phi $, where $ \phi_i $ and $ \phi_{i+1} $ represent the boundary of the $ i^{\text{th}} $ quantization region. The selected index, $ i $, is then transmitted back to the transmitter over the error-free quantized feedback channel. Therefore, CSI is partially known to the transmitter. 

Next, a transmission rate $ r_{i} $ is selected according to $ i $. In this paper, we study the optimum feedback scheme that maximizes of the goodput of the communication system, which is the overall correctly received information rate (in bits-per-second). For instance, for a fixed-rate communication system, i.e. $ \Phi = 1 $, with error probability $ \epsilon $ the goodput of the communication system can be computed as 
\begin{equation}\label{key}
r(1-\epsilon)
\end{equation}
where $ r $ is the transmission rate and $ (1-\epsilon) $ represents the success rate of the communication system. Thus, the goal of the study is to investigate the optimum selections of $ r_{i} $'s and $ \phi_i $'s such that the long-term goodput of the communication system is maximized. 


A power budget constraint is allocated on the power of the codeword such that
\begin{equation}\label{eq_power_budget_constraint}
\lVert \boldsymbol{x} \rVert^2 \leq P.
\end{equation}
It is clear that equality in Eq. \eqref{eq_power_budget_constraint} yields the maximum goodput since there is no cost on transmission rate is incurred with increasing the power up to the upper limit.

Next, previously obtained results under the assumption of infinite blocklength are first introduced and then our contributions are presented.

\begin{figure}[t]
	\centering
	\whencolumns{
		\includegraphics[width=.7\linewidth]{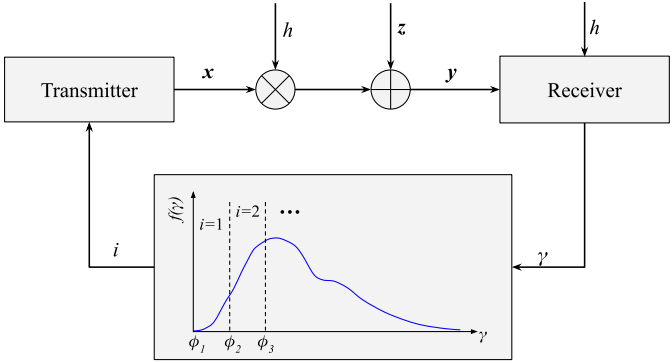}
	}{
		\includegraphics[width=1\linewidth]{system_model.eps}
	}
	\caption{System model}
	\label{fig_system_model}
\end{figure}

\section{Infinite Blocklength Analysis}

It is shown by Shannon in \cite{shannon_a_mathematical} that there exists a codebook describing the collection of codewords, with size $ 2^{nC}$ codewords, such that the codeword error probability (CEP) vanishes as $ n \rightarrow \infty $, where CEP is the probability that the receiver decides in favor of a codeword that is different from the one actually sent and $ C $ is the channel capacity and represents the maximum transmission rate for which error-free communication can be achieved. 

The goodput of the communication system can be maximized by increasing the transmission rate up to the channel capacity and the channel capacity can be maximized by selecting the maximum allowed power $ P $. Thus, for a given channel gain realization $ \gamma $, the instantaneous channel capacity can be computed as
\begin{equation}\label{eq_shannon_cap}
C(\gamma) = \log\left( 1+P\gamma \right) .
\end{equation}
When perfect CSI is available at the transmitter, it is possible to match the transmission rate according to the instantaneous channel capacity, i.e. $ C(\gamma) $. In this case, the maximum achievable goodput is the ergodic capacity of the channel which is
\begin{equation}\label{eq_inf_block_erg_cap}
\int_{0}^{\infty} f(\gamma) C(\gamma) \mathrm{d}\gamma .
\end{equation}
On the other hand, when no CSI is present at the transmitter, the maximum achievable goodput can be found by solving the following optimization problem
\begin{equation}\label{eq_inf_block_fixed_rate_cap}
\underset{r\geq 0}{\text{maximize}} ~~ \int_{\frac{1}{P}(2^r -1)}^{\infty} rf(\gamma) \mathrm{d}\gamma .
\end{equation}
Notice that Eq. \eqref{eq_inf_block_erg_cap} and Eq. \eqref{eq_inf_block_fixed_rate_cap} provide the best and the worst performance bounds, i.e. perfect and no CSI at the transmitter, to the quantized feedback scheme.

Let the transmitter select a capacity-achieving codebook according to the feedback index $ i $, such that\footnote{All logarithms in this paper are with base 2.}
\begin{equation}\label{key}
r_i = \log\left( 1+P\gamma_i \right) ,
\end{equation}
where $ \gamma_i \in [\phi_i, \phi_{i+1}] $ represents the selected reconstruction point of the $ i $th quantization region. Now, suppose that a given channel realization $ \gamma $ belongs to the $ i $th quantization region, i.e. $ \gamma \in [\phi_i, \phi_{i+1}) $. We know that as long as $ r_i \leq C(\gamma) $, error-free transmission is possible. Thus, successful decoding is only possible when $ \gamma \geq \gamma_i $. If $ \gamma < \gamma_i $, the system is in outage. Once $ \gamma_i $, $ \forall i \in \{1,2,\cdots,\Phi\} $, are selected, the outage probability of the system can be computed as
\begin{equation}\label{key}
\sum_{i=1}^{\Phi} F(\gamma_i) - F(\phi_i) .
\end{equation} 
Notice that selecting $ \gamma_i = \phi_i $ yields outage-free transmission. However, goodput maximization is not guaranteed with this selection. 

We, next, formulate the optimization problem on designing the feedback scheme that maximizes the goodput in the infinite blocklength regime. The optimal feedback scheme is the solution of the following optimization problem
\begin{subequations}
	\begin{align}\label{eq_opt_prob_infinite_block}
	\underset{\phi_i, \gamma_i, r_i}{\text{maximize}} ~&~ \sum_{i=1}^{\Phi} r_i \big(F(\phi_{i+1}) - F(\gamma_i)\big) 
	\\
	\text{subject to} ~&~ \phi_i \leq \gamma_i ,
	\\
	&~ \gamma_i \leq \phi_{i+1} .
	\end{align}
\end{subequations}
This optimization problem has been solved in \cite{kim_on_the} and \cite{makki_on_hybrid}. It is shown that the optimum can be achieved by setting
\begin{equation}\label{key}
r_i = \log\left( 1+P\phi_i \right) 
\end{equation}
and the reconstruction points are found by an iterative algorithm that can solve the following equation
\whencolumns{
	\begin{equation}\label{eq_quant_region_iterative_solution}
	F(\phi_{i+1}) = F(\phi_i) + \frac{1}{P} f(\phi_i) (1+\phi_i P) \log\left( \frac{1+\phi_i P}{1+\phi_{i-1} P} \right), ~~~ \text{for}~ i = 1,2,\cdots,\Phi,
	\end{equation}
}{
	\begin{multline}\label{eq_quant_region_iterative_solution}
	F(\phi_{i+1}) = F(\phi_i) + \frac{1}{P} f(\phi_i) (1+\phi_i P) \log\left( \frac{1+\phi_i P}{1+\phi_{i-1} P} \right),
	\\
	\text{for}~ i = 1,2,\cdots,\Phi,
	\end{multline}
}
where the first quantization region does not necessarily start from zero and therefore a new variable $ \phi_{0} = 0 $ is introduced. Thus, the overall outage probability of the optimum feedback scheme is $ F(\phi_1)$ since data transmission with zero outage probability occurs in all quantization regions except the region between $(\phi_0,\phi_1)$. When the channel power coincides within this region, the communication system is in outage and data is lost. Notably, the equation in \eqref{eq_quant_region_iterative_solution} can be achieved by applying the Karish-Kuhn-Tucker (KKT) conditions to the optimization problem. An iterative algorithm is proposed in \cite{kim_on_the} to find the optimum reconstruction points. 

\section{Finite Blocklength Analysis}

Now assume a blocklength constrained communication system where $ n $ is not allowed to take arbitrarily large values. In information theoretic studies, this is called the finite blocklength regime and, in this regime, $ C(\gamma) $ is not a correct measure to calculate the maximal allowed transmission rate. In this case, the achievable transmission rate with codewords of length $ n $ with non-zero CEP $ \epsilon $ can be closely approximated by the normal approximation, which can be computed as
\begin{equation}\label{eq_normal_approximation}
R(n,\gamma,\epsilon)=C(\gamma)-\sqrt{\frac{V(\gamma)}{n}} Q^{-1}(\epsilon)+O\left(\frac{\log n}{n}\right).
\end{equation}
where the quantity $ V(\gamma) $ is called the channel dispersion and, for AWGN channel, it is obtained as
\begin{equation}\label{eq_channel_capacity}
V(\gamma) = \left( 1 - \frac{1}{(1+\gamma P)^2}\right) (\log \exp (1))^2 ,
\end{equation} 
and $ Q^{-1}(\cdot) $ is the inverse of the Gaussian $ Q- $function
\begin{equation}\label{key}
Q(t) = \frac{1}{\sqrt{2\pi}} \int_{t}^{\infty} \exp \left(-\frac{u^2}{2}\right) \mathrm{d}u .
\end{equation}
The expression in \eqref{eq_normal_approximation} is termed as the normal approximation to the maximal transmission rate in the finite blocklength regime. Notice that $ R(n,\gamma,\epsilon) $ increases  unboundedly as $ \epsilon \rightarrow 1 $. However, with strict CEP requirements, i.e., small $ \epsilon $, the second term in Eq. \eqref{eq_normal_approximation} introduces a back-off from $ C(\gamma) $ to ensure the transmission achieves the reliability requirements with $ n $ blocklength. Thus, the achievable rate decreases substantially. On the other hand, the achievable rate increases monotonically as the blocklength $ n $ increases and converges to Shannon's capacity formula in the case of asymptotically long blocklength, i.e., $ n \rightarrow \infty $. 

Although Eq. \eqref{eq_normal_approximation} gives a close approximation on the limit of achievable rates, it does not tell how to achieve it \cite{celebi_latency_and, celebi_a_multi}. In the recent years, several encoding and decoding schemes have been proposed that can perform close to Eq. \eqref{eq_normal_approximation}. Their performance comparisons are depicted in \cite{celebi_wireless_comm, liva_code_design, shivarnimoghaddam_short_block}. In this paper, it is assumed that the transmitter selects finite-length codewords from a codebook that can achieve the achievable rates given by Eq. \eqref{eq_normal_approximation}.   

For a given transmission rate $ r $, the probability of decoding errors can then be found by solving Eq. \eqref{eq_normal_approximation} for $ \epsilon $
\begin{equation}\label{eq_cep_bound}
\epsilon = \Omega (\gamma , r ) ,
\end{equation}
where 
\begin{equation}\label{key}
\Omega (\gamma , r ) = Q \left( \big(C(\gamma)-r\big) \sqrt{\frac{n}{V(\gamma)}} \right)
\end{equation}
The behavior of $ \Omega (\gamma , r ) $ is depicted in Fig. \ref{fig_CEP_bound} for several blocklength values. As expected, in the asymptotic regime, i.e., $ n \rightarrow \infty $, error-free transmission can be achieved when $ r\leq C $. On the other hand, if $ r>C $, no information can be sent. Unlike the asymptotic regime, in the finite blocklength regime, information transmission becomes possible even for transmission rates higher than the capacity, albeit at the expense of high CEP values. On the other hand, even for rates lower than the capacity, CEP may not vanish. This result shows that the optimal feedback scheme in the asymptotic regime may not correspond to the optimal selections for finite blocklength.   

\begin{figure}[t]
	\centering
	\whencolumns{
		\includegraphics[width=.6\linewidth]{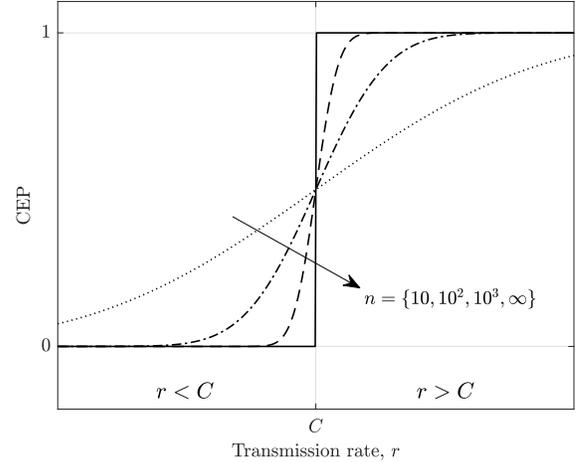}
	}{
		\includegraphics[width=1\linewidth]{cep_vs_r.eps}
	}
	\caption{Differences between the CEP characteristics in infinite and finite blocklength regimes.}
	\label{fig_CEP_bound}
\end{figure}

\subsection{Finite Blocklength Analysis in Quasi-Static Channels}

The instantaneous capacity of a channel, given a channel gain $ \gamma $, can be calculated as in Eq. \eqref{eq_shannon_cap}. Assuming that a fixed transmission rate $ r $ is selected, there is a certain probability that the transmission is in outage and reliable communcation is not possible. The outage probability of a communication system can be calculated as
\begin{equation}\label{key}
\epsilon^{\text{out}} = \mathbb{P}\{\log(1+P \gamma ) < r\} ,
\end{equation}
where the probability is over the channel gain $ \gamma $. Thus, the outage capacity or the $ \epsilon $-capacity of a communication system can be found as the largest $ r $ that fulfills $ \mathbb{P}\{\log(1+\gamma\mathcal{P}) < r\} \leq \epsilon $:
\begin{equation}\label{eq_outage_cap_inf}
C^{\epsilon^{\text{out}}} = \sup \{ r\geq 0: ~ \mathbb{P}\{\log(1+ P \gamma) < r\} \leq \epsilon\}
\end{equation}
Notice that the outage capacity $ C^{\epsilon^{\text{out}}} $ reveals the maximum transmission rate such that the CEP of the communication system is $ \epsilon $, assuming $ n \rightarrow \infty $.

The performance of finite blocklength codes in quasi-static fading channels is analyzed in \cite{yang_block_fading} and \cite{yang_quasi_static}, and it is shown that the channel dispersion $ V(\gamma) $ is zero for many classes of channel models such that the achievable rate can be closely approximated by the outage capacity 
\begin{equation}\label{eq_achiev_rate_out}
R\left(n,\gamma,\epsilon^{\text{out}}\right) = C^{\epsilon^{\text{out}}} + O\left(\frac{\log n}{n}\right) .
\end{equation}
Notice that since the dispersion is zero, no additional term on $ V(\gamma)/n $ appears in Eq. \eqref{eq_achiev_rate_out}. This shows that $ R(n,\gamma,\epsilon) $ in Eq. \eqref{eq_achiev_rate_out} converges to the outage capacity $ C^{\epsilon^{\text{out}}} $ faster than $ R(n,\gamma,\epsilon) $ in Eq. \eqref{eq_normal_approximation} converges to the capacity $ C(\gamma) $. Thus, one can conclude that the blocklength in the finite blocklength regime under quasi-static fading channels has smaller impact on the achievable rate.

However, this result does not hold for all types of channels. Consider, for example, a Rician fading channel with $h \sim \mathcal{CN}(\mu, \sigma^2)$. In case the line-of-sight component $\mu$ remains constant and the variance $\sigma^2$ of the non-LoS component converges to zero, the channel approaches to an AWGN channel, where the maximum achievable rate is given by Eq. \eqref{eq_normal_approximation}. In that case, the convergence to the capacity occurs on the order of $1/\sqrt{n}$, not $(\log n)/n$. Therefore, the expression in \eqref{eq_achiev_rate_out} does not hold for such a channel. 

To extend the normal approximation in \eqref{eq_normal_approximation} for different types of quasi-static fading channels, the following expression is proposed in \cite{yang_quasi_static}
\begin{equation}\label{eq_general_cep_formula}
\epsilon = \mathbb{E}\big[ \Omega (\gamma , r )  \big],
\end{equation}
where the expectation is taken with respect to the distribution of the random channel gain $ \gamma $. Hence, the overall success rate of the communication system can be calculated as
\begin{equation}\label{eq_overall_success_rate_fin}
1-\epsilon =  \int_{0}^{\infty} f(\gamma) \big( 1 - \Omega (\gamma , r )  \big) \mathrm{d}\gamma.
\end{equation}


\section{Goodput Maximization with Quantized Feedback}

The overall success rate given in Eq. \eqref{eq_overall_success_rate_fin} is achieved when the transmitter selects a fixed rate. A better performance can be achieved if the transmitter has knowledge on the channel gain. In this section, we investigate the goodput in finite blocklength regime for the system model presented in Fig. \ref{fig_system_model}.

Given a transmission rate $ r_i $ and by using Eq. \eqref{eq_overall_success_rate_fin}, the goodput of the quantization region between $ \phi_i $ and $\phi_{i+1}$ yields 
\begin{equation}\label{eq_exp_goodput}
r_i \int_{\phi_i}^{\phi_{i+1}} f(\gamma) \big(1-\Omega(\gamma,r_i)\big) \mathrm{d}\gamma .
\end{equation}
The total goodput of the communication system can be calculated by summing Eq. \eqref{eq_exp_goodput} over all $ i $'s. The integral in Eq. \eqref{eq_exp_goodput} can be decomposed as 
\whencolumns{
	\begin{equation}\label{eq_obj_func_sep}
	\int_{\phi_i}^{\phi_{i+1}} f(\gamma) \big( 1 - \Omega (\gamma , r_i )  \big) \mathrm{d}\gamma =
	F(\phi_{i+1}) - F(\phi_i) - \int_{\phi_i}^{\phi_{i+1}} f(\gamma) \Omega (\gamma , r_i )   \mathrm{d}\gamma
	\end{equation}
}{
	\begin{multline}\label{eq_obj_func_sep}
	\int_{\phi_i}^{\phi_{i+1}} f(\gamma) \big( 1 - \Omega (\gamma , r_i )  \big) \mathrm{d}\gamma =
	\\
	F(\phi_{i+1}) - F(\phi_i) - \int_{\phi_i}^{\phi_{i+1}} f(\gamma) \Omega (\gamma , r_i )   \mathrm{d}\gamma
	\end{multline}
}
Notably, the objective function defined for infinite blocklength regime is equivalent to the expression above without the third term. The third term in Eq. \eqref{eq_obj_func_sep} introduces a backoff due to the finite blocklength. Notice that setting $ r_i = 0 $ yields the goodput in Eq. \eqref{eq_exp_goodput} to be zero. As $ r_i $ increases, Eq. \eqref{eq_exp_goodput} increases up to some value but then, it starts to decrease due to the third term since it can be shown that
\begin{equation}\label{key}
\lim\limits_{r_i \rightarrow \infty} \int_{\phi_i}^{\phi_{i+1}} f(\gamma)  \Omega(\gamma,r_i)    \mathrm{d}\gamma = F(\phi_{i+1}) - F(\phi_i) ,
\end{equation}
which reduces the goodput to zero.

To demonstrate the differences between the achievable goodput values with infinite and finite blocklength, we show a case study where $ h $ is Rician with $ K=10 $, $ K $ being the ratio between the power in the direct path and the power in the scattered paths, and therefore $ \gamma $ is a noncentral-$ \chi^2 $ distributed random variable with 
\begin{equation}\label{key}
f_{\chi^2}(\gamma) = (K+1) \exp\big( -\gamma(K+1)+K \big) I_0\left( 2\sqrt{\gamma K(K+1)} \right) ,
\end{equation}
where $ I_0(\cdot) $ is the zero-order modified Bessel function of the first kind. We select $ \Phi = 3 $ and the quantization regions are $ \phi_i = \{0, 0.9, 1.2, \infty\} $, not necessarily being the optimum selections. The achievable goodput values for each quantization region with respect to $ r_i \in [0, 5]$ are depicted in Fig. \ref{fig_diff_bw_quant_regions_for_inf_fin}. More specifically, the black solid lines and dashed red lines represent the results achieved with Eq. \eqref{eq_opt_prob_infinite_block} and Eq. \eqref{eq_exp_goodput}, respectively. Maximums are highlighted with square and circle markers. 

\begin{figure}[t]
	\centering
	\whencolumns{
		\includegraphics[width=.6\linewidth]{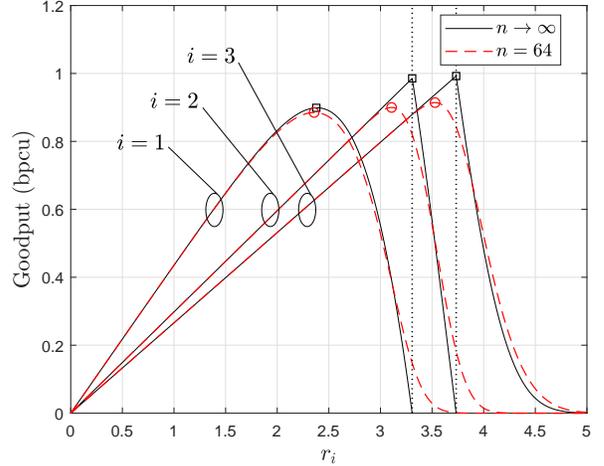}
	}{
		\includegraphics[width=1\linewidth]{diff_bw_quant_regions_for_inf_fin.eps}
	}
	\caption{An illustration of the achievable goodput values with infinite and finite blocklengths for $ \Phi = 3 $ where $ \phi_i = \{0, 0.9, 1.2, \infty\} $.}
	\label{fig_diff_bw_quant_regions_for_inf_fin}
\end{figure}

Fig. \ref{fig_diff_bw_quant_regions_for_inf_fin} reveals several important findings. It can be seen that there are significant differences between the maximum achievable goodput values for finite and infinite blocklength regimes at each quantization region. Due to the backoff, maximum achievable goodput is lower in finite blocklength for all quantization regions. In addition, one can end up with even lower goodput values if the optimal scheme for asymptotic regime is applied for finite blocklength. On the other hand, this difference is slightly smaller for $ i=1 $ but increases as $ i $ increases. Similar empirical results are observed in various different channel types. Thus, the impact of finite blocklength is smaller for quantization regions with lower channel power, i.e. $ i=1 $, whereas it increases for quantization regions with higher channel power, i.e. $ i>1 $. This is due to the fact that the channel dispersion, $ V(\gamma) $, is shown to be zero when the channel experiences deep fading, which is the case when $ i=1 $. However, when $ i>1 $, since the channel gain becomes higher, the channel converges to an AWGN channel, i.e. $ V(\gamma) \neq 0 $, and the impact of finite blocklength becomes extremely important.\footnote{A more thorough analysis is presented in Fig. \ref{fig_benchmark_comparison} in Sec. VII.}   


\subsection{Problem Formulation}

The optimization problem on the maximization of the goodput in finite blocklength regime can be formulated as 
\begin{subequations}
	\begin{align}
		\underset{r_i,\phi_i}{\text{maximize}} ~~& \sum_{i=1}^{\Phi} r_i \int_{\phi_i}^{\phi_{i+1}} f(\gamma) \big(1-\Omega(\gamma,r_i)\big) \mathrm{d}\gamma 
		\label{eq_opt_prob_1_obj}
		\\
		\text{s.t.} ~~& r_i \geq 0, \label{eq_opt_prob_1_cons_1}
		\\
		& \phi_{i+1} \geq \phi_i, \label{eq_opt_prob_1_cons_2}
	\end{align}
\end{subequations}
where the maximization is done over the variables $ r_i $ and $ \phi_i $ with the constraints that are listed in Eq. \eqref{eq_opt_prob_1_cons_1} and Eq. \eqref{eq_opt_prob_1_cons_2}, respectively.\footnote{Notice that no constraint on the CEP is introduced in this optimization problem. Therefore this problem may not be suitable for URLLC applications where stringent requirements on reliability is introduced. Instead, this optimization problem is more appropriate for low-latency applications. However, we first investigate this optimization problem and present the results for the completeness of the study.}

The proposed optimization problem is difficult to solve since $ r_i $ and $ \phi_i $ are coupled together. To circumvent this difficulty, we introduce a block coordinate descent based iterative solution where we optimize one set of variables while keeping the others fixed, and vice versa \cite{beck_on_the, hong_a_unified}. Then, each sub-problem is solved in an iterative manner. Below, the sub-optimization problems are introduced 

Optimization problem 1:
	\begin{align}
	\underset{r_i}{\text{maximize}} ~~\sum_{i=1}^{\Phi} r_i \int_{\phi_i}^{\phi_{i+1}} f(\gamma) \big(1-\Omega(\gamma,r_i)\big) \mathrm{d}\gamma ~~~~~  \text{s.t.} ~~ \eqref{eq_opt_prob_1_cons_1}  . \label{eq_opt_prob_2_obj}
	\end{align}

Optimization problem 2:
	\begin{align}
	\underset{\phi_i}{\text{maximize}} ~~\sum_{i=1}^{\Phi} r_i \int_{\phi_i}^{\phi_{i+1}} f(\gamma) \big(1-\Omega(\gamma,r_i)\big) \mathrm{d}\gamma ~~~~~  \text{s.t.} ~~ \eqref{eq_opt_prob_1_cons_2} . \label{eq_opt_prob_3_obj}
	\end{align}

Next, we investigate these two optimization problems.

\subsection{Problem 1: Optimization of Transmission Rate}

\begin{theorem}
	For fixed $ \phi_i $, the optimum $ r_i $ can be found by solving the following equation
	\begin{equation}\label{eq_theorem_opt_r_i}
	r_i = -\frac{\int_{\phi_i}^{\phi_{i+1}} f(\gamma) \big(1-\Omega(\gamma,r_i)\big) \mathrm{d}\gamma}{ \int_{\phi_i}^{\phi_{i+1}}H(\gamma,r_i) \mathrm{d}\gamma} ,
	\end{equation}
	where
	\begin{equation}\label{key}
	H(\gamma,r_i) = -\frac{1}{\sqrt{2\pi}} \exp\left(-\frac{n\big(C(\gamma)-r_i\big)^2}{2V(\gamma)}\right) f(\gamma)\sqrt{\frac{n}{V(\gamma)}} .
	\end{equation}
\end{theorem}

\begin{proof}
	We first investigate the convexity of the objective function with respect to $ r_i $. The second-order derivative of the objective function follows
	\whencolumns{
		\begin{equation} \label{eq_obj_func_second_deriv}
		\frac{\partial^2}{\partial r_i^2} r_i \int_{\phi_i}^{\phi_{i+1}} f(\gamma) \big(1-\Omega(\gamma,r_i)\big) \mathrm{d}\gamma =  2 \int_{\phi_i}^{\phi_{i+1}} H(\gamma,r_i) \mathrm{d}\gamma + r_i \int_{\phi_i}^{\phi_{i+1}} H(\gamma,r_i) \frac{n\big(C(\gamma)-r_i\big)}{V(\gamma)} \mathrm{d}\gamma
		\end{equation}
	}{
		\begin{multline} \label{eq_obj_func_second_deriv}
			\frac{\partial^2}{\partial r_i^2} r_i \int_{\phi_i}^{\phi_{i+1}} f(\gamma) \big(1-\Omega(\gamma,r_i)\big) \mathrm{d}\gamma = 
			\\
			2 \int_{\phi_i}^{\phi_{i+1}} H(\gamma,r_i) \mathrm{d}\gamma + r_i \int_{\phi_i}^{\phi_{i+1}} H(\gamma,r_i) \frac{n\big(C(\gamma)-r_i\big)}{V(\gamma)} \mathrm{d}\gamma
		\end{multline}
	}
	where we used the following equality
	\begin{equation}\label{key}
	\frac{\partial}{\partial t} Q(t) = -\frac{1}{\sqrt{2\pi}} \exp\left(-\frac{t^2}{2}\right).
	\end{equation}

	Since $ n>0 $, $ C(\gamma)>0 $, $ V(\gamma)>0 $, and, by definition, $ f(\gamma)>0 $, and, however, $ H(\gamma,r) $ is always negative. Therefore, the first term in Eq. \eqref{eq_obj_func_second_deriv} is always negative. Similar to the first term, the second term is also negative as long as $ C(\gamma)>r_i $. However, depending on the interval $ \phi_i $ and $ \phi_{i+1} $, the second integral becomes positive as $ r_i $ increases. Thus, Eq. \eqref{eq_obj_func_second_deriv} is negative between $ r_i\in[0,r_i'] $ and the objective function is concave and for higher values of $ r_i $, Eq. \eqref{eq_obj_func_second_deriv} is positive and the objective is convex, where $ r_i' $ is the point of inflection and can be found by solving the following equation 
	\begin{multline}\label{key}
	\bigg[ 2 \int_{\phi_i}^{\phi_{i+1}} H(\gamma,r_i) \mathrm{d}\gamma 
	\\
	+ r_i \int_{\phi_i}^{\phi_{i+1}} H(\gamma,r_i) \frac{n\big(C(\gamma)-r_i\big)}{V(\gamma)} \mathrm{d}\gamma \bigg]_{r_i = r_i'} = 0.
	\end{multline}
	We call this function a concave-convex function \cite{cai_stocastic}. 

	Since the function is concave-convex and since it tends to zero as $ r_i=0 $ and $ r_i \rightarrow \infty $, one can show that the maximum always lies within the concave region. Therefore, the optimum value of $ r_i $ that maximizes the objective function can be found by taking the first derivative, setting it to zero, and solving it for $ r_i $, which yields Eq. \eqref{eq_theorem_opt_r_i}.
\end{proof}

Although it is possible to numerically evaluate optimum $ r_i $, no closed form expression can be achieved due to the integrals. Next, we derive an approximation on $ r_i $.

\begin{theorem}
	A close appoximation on the optimal $ r_i $ can be achieved with the following equation
	\begin{multline}\label{eq_theorem_1}
	\Big[ F(\min\{\Delta_2, \phi_{i+1}\}) - F(\max\{\Delta_1, \phi_i\}) \Big] u(\phi_{i+1}-\Delta_1) u(\Delta_2-\phi_i)
	\\
	\times \left(\frac{1}{2} + b\Big(a-\mathbb{E}\big[\gamma\big|\max\{\Delta_1, \phi_i\} \leq \gamma \leq  \min\{\Delta_2, \phi_{i+1}\}\big]\Big)\right) 
	\\
	+ r_i \Big[ F(\phi_{i+1}) - F(\max\{\Delta_2, \phi_i\}) \Big] u(\phi_{i+1} - \Delta_2 ) 
    \\
	= - r_i \int_{\phi_i}^{\phi_{i+1}} H(\gamma,r_i) \mathrm{d}\gamma .
	\end{multline}
	where $ u(\cdot) $ represents the unit step function and $\Delta_1 = \frac{1}{2b} + a$, $\Delta_2 = -\frac{1}{2b} + a$ with $a  = \frac{1}{P}(\exp(r_i)-1)$ and $b = -P \sqrt{\frac{n}{2\pi(\exp(2r_i)-1)}}$.
\end{theorem}

\begin{proof}
	The objective function for single $ i $ can be approximated by linearizing the $ \Omega(\cdot) $ function at the point $ \gamma  = a $. The approximation of the $ \Omega(\cdot) $ function can be written as \cite{makki_finite_block}
	\begin{align}\label{eq_q_func_approx}
	\Omega(\gamma,r_i)  \approx 
	\begin{cases}
	1 &, ~ \gamma < \Delta_1,
	\\
	1/2 + b(\gamma - a) &, ~ \gamma \in \left[\Delta_1, \Delta_2\right],
	\\
	0 &,~ \gamma > \Delta_2,
	\end{cases}
	\end{align}
	where $ b $ is the derivative of the $ Q(\cdot) $ function at $ \gamma = a $, that is 
	\begin{equation}\label{key}
	b = \frac{\partial}{\partial \gamma} \Omega(\gamma,r_i) ~ \bigg|_{\gamma = a} 
	= -P \sqrt{\frac{n}{2\pi(\exp(2r_i)-1)}}.
	\end{equation} 
	
	Implementing $ a $ and $ b $ into Eq. \eqref{eq_q_func_approx} and solving the integral yields
	\whencolumns{
		\begin{equation}\label{eq_approx_long_version}
		\int_{\phi_i}^{\phi_{i+1}} f(\gamma) \big(1-\Omega(\gamma,r_i)\big) \mathrm{d}\gamma \approx
		\begin{cases} 
		0, ~ \text{if} ~ \phi_i \leq \phi_{i+1} \leq \Delta_1,
		\\
		r_i \left(\frac{1}{2}+ba-b \mathbb{E}\left[\gamma | \Delta_1\leq\gamma\leq \phi_{i+1}\right]\right)\left(F(\phi_{i+1})-F(\Delta_1)\right), ~ \\ ~~~~~~~~~~~~~~~~~~~~~~~~~~ \text{if} ~ \phi_i \leq \Delta_1 ~\text{and}~ \Delta_1 \leq \phi_{i+1} \leq \Delta_2 ,
		\\
		r_i\left(F(\phi_{i+1}-F(\Delta_2))\right)+r_i \left(\frac{1}{2}+ba-b \mathbb{E}\left[\gamma | \Delta_1\leq\gamma\leq\Delta_2\right]\right)
		\\
		~~~~\times\left(F(\Delta_2)-F(\Delta_1)\right), ~  \text{if} ~ \phi_i \leq \Delta_1 ~\text{and}~ \Delta_2 \leq \phi_{i+1},
		\\
		r_i \left(\frac{1}{2}+ba-b \mathbb{E}\left[\gamma | \phi_i\leq\gamma\leq \phi_{i+1}\right]\right)\left(F(\phi_{i+1})-F(\phi_{i})\right), ~ \\ ~~~~~~~~~~~~~~~~~~~~~~~~~~ \text{if} ~ \Delta_1 \leq \phi_{i} \leq \phi_{i+1} \leq  \Delta_2,
		\\
		r_i\left(F(\phi_{i+1})-F(\Delta_2)\right) + r_i \left(\frac{1}{2}+ba-b \mathbb{E}\left[\gamma | \phi_i\leq\gamma\leq \Delta_2\right]\right)
		\\
		~~~~\times\left(F(\Delta_2)-F(\phi_{i})\right), ~ \text{if} ~ \Delta_1 \leq \phi_{i} \leq  \Delta_2 \leq \phi_{i+1} ,
		\\
		r_i\left(F(\phi_{i+1})-F(\phi_i)\right) , ~ \text{if} ~ \Delta_2 \leq \phi_{i} \leq \phi_{i+1} ,
		\end{cases}
		\end{equation}
	}{
		\begin{align}\label{eq_approx_long_version}
		& \int_{\phi_i}^{\phi_{i+1}}  f(\gamma) \big(1-\Omega(\gamma,r_i)\big) \mathrm{d}\gamma \approx \nonumber
		\\ & ~~
		\begin{cases} 
		0, ~ \text{if} ~ \phi_i \leq \phi_{i+1} \leq \Delta_1,
		\\
		r_i \left(\frac{1}{2}+ba-b \mathbb{E}\left[\gamma | \Delta_1\leq\gamma\leq \phi_{i+1}\right]\right)\left(F(\phi_{i+1})-F(\Delta_1)\right), ~ \\ ~~~~~~~~~~~~~~~~~~~~~~~~~~ \text{if} ~ \phi_i \leq \Delta_1 ~\text{and}~ \Delta_1 \leq \phi_{i+1} \leq \Delta_2 ,
		\\
		r_i\left(F(\phi_{i+1}-F(\Delta_2))\right)+r_i \left(\frac{1}{2}+ba-b \mathbb{E}\left[\gamma | \Delta_1\leq\gamma\leq\Delta_2\right]\right)
		\\
		~~~~\times\left(F(\Delta_2)-F(\Delta_1)\right), ~  \text{if} ~ \phi_i \leq \Delta_1 ~\text{and}~ \Delta_2 \leq \phi_{i+1},
		\\
		r_i \left(\frac{1}{2}+ba-b \mathbb{E}\left[\gamma | \phi_i\leq\gamma\leq \phi_{i+1}\right]\right)\left(F(\phi_{i+1})-F(\phi_{i})\right), ~ \\ ~~~~~~~~~~~~~~~~~~~~~~~~~~ \text{if} ~ \Delta_1 \leq \phi_{i} \leq \phi_{i+1} \leq  \Delta_2,
		\\
		r_i\left(F(\phi_{i+1})-F(\Delta_2)\right) + r_i \left(\frac{1}{2}+ba-b \mathbb{E}\left[\gamma | \phi_i\leq\gamma\leq \Delta_2\right]\right)
		\\
		~~~~\times\left(F(\Delta_2)-F(\phi_{i})\right), ~ \text{if} ~ \Delta_1 \leq \phi_{i} \leq  \Delta_2 \leq \phi_{i+1} ,
		\\
		r_i\left(F(\phi_{i+1})-F(\phi_i)\right) , ~ \text{if} ~ \Delta_2 \leq \phi_{i} \leq \phi_{i+1} ,
		\end{cases}
		\end{align}
	}
	where 
	\begin{equation}\label{key}
	\mathbb{E}\left[\gamma | \gamma_1\leq\gamma\leq\gamma_2\right] = \frac{1}{F(\gamma_2)-F(\gamma_1) } \int_{\gamma_1}^{\gamma_2}\gamma f(\gamma) \mathrm{d}\gamma.
	\end{equation}
	represents the expected value of the truncated random variable $ \gamma $ such that $ \gamma $ can only have values between $ \gamma_1\leq\gamma\leq\gamma_2 $. Merging all cases in Eq. \eqref{eq_approx_long_version} leads to
	\whencolumns{
		\begin{multline} \label{eq_q_func_int_approx}
		\int_{\phi_i}^{\phi_{i+1}} f(\gamma) \big(1-\Omega(\gamma,r_i)\big) \mathrm{d}\gamma \approx 
		\Big[ F(\phi_{i+1}) - F(\max\{\Delta_2, \phi_i\}) \Big] u(\phi_{i+1} - \Delta_2)
		\\
		+ \Big[ F(\min\{\Delta_2, \phi_{i+1}\}) - F(\max\{\Delta_1, \phi_i\}) \Big]  u(\phi_{i+1}-\Delta_1) u(\Delta_2-\phi_i) 
		\\
		\times \left( \frac{1}{2} + b\Big(a-\mathbb{E}\big[\gamma\big|\max\{\Delta_1, \phi_i\} \leq \gamma \leq  \min\{\Delta_2, \phi_{i+1}\}\big]\Big) \right) .
		\end{multline}
	}{
		\begin{multline} \label{eq_q_func_int_approx}
		\int_{\phi_i}^{\phi_{i+1}} f(\gamma) \big(1-\Omega(\gamma,r_i)\big) \mathrm{d}\gamma \approx 
		\\
		\Big[ F(\min\{\Delta_2, \phi_{i+1}\}) - F(\max\{\Delta_1, \phi_i\}) \Big]  u(\phi_{i+1}-\Delta_1) u(\Delta_2-\phi_i)
		\\
		\times \left( \frac{1}{2} + b\Big(a-\mathbb{E}\big[\gamma\big|\max\{\Delta_1, \phi_i\} \leq \gamma \leq  \min\{\Delta_2, \phi_{i+1}\}\big]\Big) \right) 
		\\
		+ \Big[ F(\phi_{i+1}) - F(\max\{\Delta_2, \phi_i\}) \Big] u(\phi_{i+1} - \Delta_2) .
		\end{multline}
	}
	Finally, setting the first order derivative of Eq. \eqref{eq_q_func_int_approx} to zero yields Eq. \eqref{eq_theorem_1}.
\end{proof}

\subsection{Problem 2: Optimization of Quantization Regions}

Now we investigate the second step of the optimization problem where $ r_i $ is fixed but we let $ \phi_i $ to vary. We then assume that a constraint qualification holds at the maximizers of (13), so that the KKT conditions are necessary optimality conditions. 

\begin{theorem}
	For fixed $ r_i $, the optimal quantization regions can be calculated by solving the following equation for $ \phi_i $ 
	\begin{equation}\label{eq_condition_for_qi}
	\frac{r_{i-1}}{r_i} = \frac{1-\Omega(\phi_i,r_{i})}{1-\Omega(\phi_i,r_{i-1})} .
	\end{equation}
\end{theorem}

\begin{proof}
	Let $ \lambda_i \geq 0 $ represents the Lagrangian multiplier associated with the CEP constraints. The Lagrangian function of \eqref{eq_opt_prob_3_obj} is
	\begin{equation}\label{key}
	\mathcal{L} = \sum_{i=1}^{\Phi} r_i \int_{\phi_i}^{\phi_{i+1}} f(\gamma) \big(1-\Omega(\gamma,r_i)\big) \mathrm{d}\gamma
	- \sum_{i=1}^{\Phi} \lambda_i (\phi_i - \phi_{i+1})
	\end{equation}
	where $ \lambda_i \geq 0 $ are the non-negative Lagrange multipliers that are associated with the constraint in Eq. \eqref{eq_opt_prob_1_cons_2}.
	
	By a direct investigation of the activeness of the linear constraint, we can write the KKT conditions as
	\whencolumns{
		\begin{align}
		& \frac{\partial}{\partial \phi_i} \mathcal{L} =  r_{i-1} \big(1-\Omega(\phi_i,r_{i-1})\big) f(\phi_i)  - r_{i} \big(1-\Omega(\phi_i,r_i)\big) f(\phi_i) -\lambda_i + \lambda_{i-1} = 0, \label{eq_kkt_derivative}
		\\
		& \lambda_i (\phi_i - \phi_{i+1}) = 0 , \label{eq_kkt_2}
		\end{align}
	}{
		\begin{align}
		\frac{\partial}{\partial \phi_i} \mathcal{L} &=  r_{i-1} \big(1-\Omega(\phi_i,r_{i-1})\big) f(\phi_i) \nonumber
		\\
		& ~~~~~~~~ - r_{i} \big(1-\Omega(\phi_i,r_i)\big) f(\phi_i) -\lambda_i + \lambda_{i-1} = 0, \label{eq_kkt_derivative}
		\\
		\lambda_i (\phi_i &- \phi_{i+1}) = 0 , \label{eq_kkt_2}
		\end{align}
	}
	where in Eq. \eqref{eq_kkt_derivative} we used the Leibniz integral rule which states that the derivative of a definite integral can be expressed as
	\whencolumns{
		\begin{equation}\label{key}
		\frac{\partial}{\partial v} \left( \int_{u_1(v)}^{u_2(v)} h(v,t) \mathrm{d}t \right) = \int_{u_1(v)}^{u_2(v)} \frac{\partial}{\partial v} h(v,t) \mathrm{d}t 
		+h(v,u_2(v)) \frac{\partial}{\partial v} u_2(v) -h(v,u_1(v)) \frac{\partial}{\partial v} u_1(v) .
		\end{equation}
	}{
		\begin{multline}\label{key}
		\frac{\partial}{\partial v} \left( \int_{u_1(v)}^{u_2(v)} h(v,t) \mathrm{d}t \right) = \int_{u_1(v)}^{u_2(v)} \frac{\partial}{\partial v} h(v,t) \mathrm{d}t 
		\\
		+h(v,u_2(v)) \frac{\partial}{\partial v} u_2(v) -h(v,u_1(v)) \frac{\partial}{\partial v} u_1(v) .
		\end{multline}
	}
	Solving the KKT conditions in Eq. \eqref{eq_kkt_derivative} and \eqref{eq_kkt_2} yields Eq. \eqref{eq_condition_for_qi}.
\end{proof}

Solving Eq. \eqref{eq_condition_for_qi} for $ \phi_i $ leads to the quantization scheme that maximizes the  overall goodput with the sub-optimum rate selections, which are also updated with every step of the iteration of the algorithm. Here, one can set $ \phi_1 = 0 $ and $ \phi_{\Phi+1} = \infty $, and Eq. \eqref{eq_condition_for_qi} needs to be solved for all $ \phi_i $, $ i=2,\cdots,\Phi $ in every iteration. Notably, an interesting outcome of Eq. \eqref{eq_condition_for_qi} is that the optimum selection of $ \phi_i $ given $ r_i $ and  $ r_{i-1} $ does not depend on $ f(\gamma) $. 


\begin{algorithm}[t]
    \small
	\caption{Goodput maximization}
	\begin{algorithmic}[1]
		\State Initialize $ \phi_i $ for $ i=1,2,\cdots, \Phi+1 $  
		\Repeat
			\For{$ i=1,2,\cdots, \Phi $} 
			\State Solve \eqref{eq_theorem_1} for $ r_i $
			\EndFor
			\For{$ i=\Phi,\Phi-1,\cdots, 2 $} 
			\State Solve \eqref{eq_condition_for_qi} for $ \phi_i $
			\EndFor
		\Until \textit{Convergence}
	\end{algorithmic} 
\end{algorithm}

The proposed algorithm is shown in Algorithm 1. It can be seen that the iterative algorithm continues to run until the quantization regions achieve a steady state in which the minimum difference between the values of the updated quantization regions at two subsequent iterations is less than some previously setted threshold value, denoted as $ \tau $.

\section{Goodput Maximization with Constraint on CEP}

The maximization problem studied in Sec. V does not impose a constraint on the CEP. However, reliability in a URLLC system is crucial and such a constraint needs to be implemented.  In this section, we jointly optimize the goodput and CEP.   

\subsection{Problem Formulation}

With a fixed rate $ r_i $, the CEP at the $ i^{\text{th}} $ quantization region can be calculated as 
\begin{equation}\label{eq_cep_ith_quant_reg}
\epsilon_i =  \int_{\phi_i}^{\phi_{i+1}} f(\gamma) \Omega(\gamma,r_i) \mathrm{d}\gamma
\end{equation}
and the overall CEP of the system is the sum of Eq. \eqref{eq_cep_ith_quant_reg} over all $ i $'s. Thus, with the reliability constraint, the optimization problem can now be formulated as 
\begingroup
\allowdisplaybreaks
\begin{subequations} 
	\begin{align}
	\underset{r_i,\phi_i, \epsilon_i}{\text{maximize}} ~~& \sum_{i=1}^{\Phi} r_i \int_{\phi_i}^{\phi_{i+1}} f(\gamma) \big(1-\Omega(\gamma,r_i)\big) \mathrm{d}\gamma  \label{eq_opt_prob_4_obj}
	\\
	\text{s.t.} ~~ & \sum_{i=1}^{\Phi} \epsilon_i \leq \varepsilon_m \label{eq_opt_prob_4_const_1}
	\\
	& \eqref{eq_opt_prob_1_cons_1}, \eqref{eq_opt_prob_1_cons_2}, \nonumber
	\end{align}
\end{subequations}
\endgroup
where Eq. \eqref{eq_opt_prob_4_const_1} specifies the maximum error probability. Notably, Eq. \eqref{eq_opt_prob_4_const_1} introduces contraint on the sum of the CEPs of the quantization regions. This constraint increases the complexity of the optimization problem such that the optimum cannot be achieved due to the intractable complexity. However, in order to relax the optimization problem, we introduce a new variable $ \varepsilon_i $, which denotes the CEP constraints on the $ i^{\text{th}} $ quantization region such that the sum of $ \varepsilon_i $'s is less than or equal to the overall CEP constraint $ \varepsilon_m $. Thus, the optimization problem follows
\begingroup
\allowdisplaybreaks
\begin{subequations}
	\begin{align}
	\underset{r_i,\phi_i,\epsilon_i, \varepsilon_i}{\text{maximize}} ~~& \sum_{i=1}^{\Phi} r_i \int_{\phi_i}^{\phi_{i+1}} f(\gamma) \big(1-\Omega(\gamma,r_i)\big) \mathrm{d}\gamma \label{eq_opt_prob_5_obj}
	\\
	\text{s.t.} ~~ & \epsilon_i \leq \varepsilon_i , \label{eq_opt_prob_5_const_1}
	\\
	& \sum_{i=1}^{\Phi} \varepsilon_i \leq \varepsilon_m, \label{eq_opt_prob_5_const_2}
	\\
	& \varepsilon_i > 0, \label{eq_opt_prob_5_const_3}
	\\
	& \eqref{eq_opt_prob_1_cons_1}, \eqref{eq_opt_prob_1_cons_2}. \nonumber
	\end{align}
\end{subequations}
\endgroup
Notice that $ \varepsilon_i $, for $ i=1,2,\cdots,\Phi $, are introduced as variables and the optimum selection of $ \varepsilon_i $ is also unknown. 

\begin{figure*}[t]
	\centering
	\begin{subfigure}{.5\textwidth}
		\centering
		\includegraphics[width=1\linewidth]{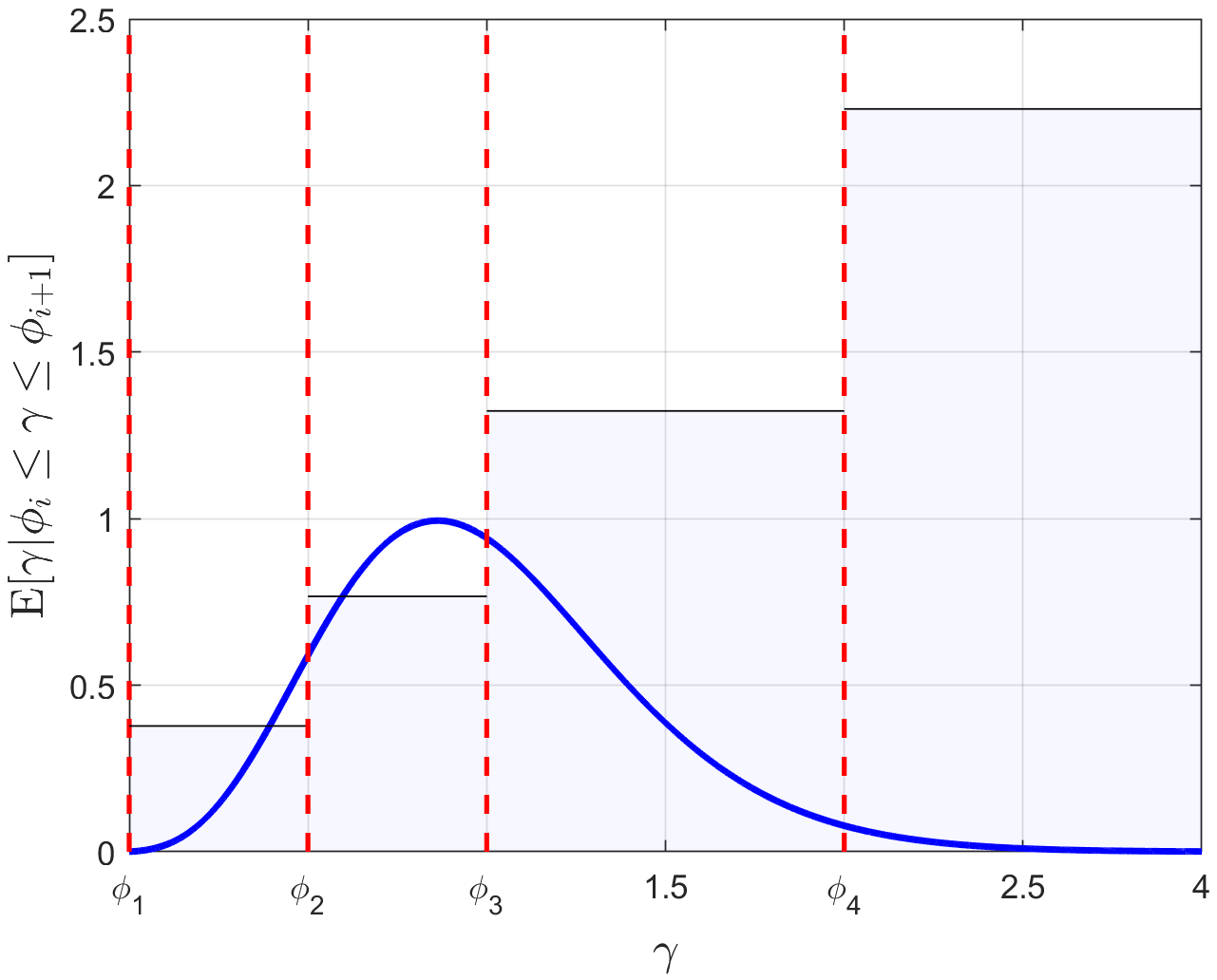}
		\caption{}
		\label{fig_trunc_mean_rician_K_10}
	\end{subfigure}%
	\begin{subfigure}{.5\textwidth}
		\centering
		\includegraphics[width=1\linewidth]{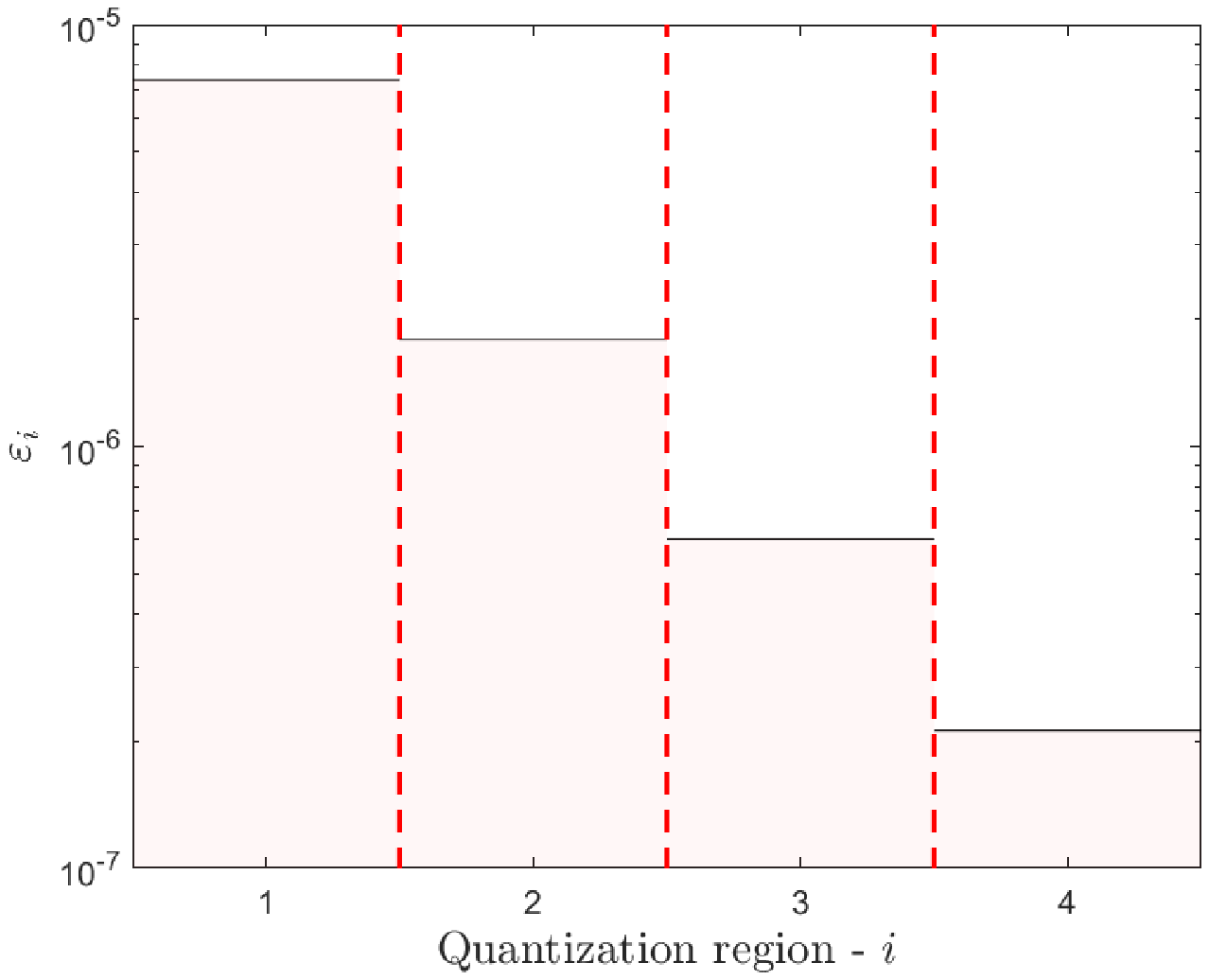}
		\caption{}
		\label{fig_cep_assignment_ricain_K10}
	\end{subfigure}
	\caption{An illustration of the proposed CEP assignment when $ \Phi = 4 $, $ h $ is Rician with $ K=10 $, and $ \phi_i = \{ 0, 0.5, 1, 2, \infty \} $ for $ i=1,\cdots,5 $. (a) Truncated mean values are depicted with respect to the quantization regions. Here, the pdf of the channel power is also shown with the blue solid line. (b) The assigned CEP values are shown with respect to the quantization regions, where $ \epsilon_m = 10^{-5} $. }
	\label{cep_assignment}
\end{figure*}

\subsection{Problem 1: CEP Allocation}
In this section, we evaluate the impact of the CEP allocation on the goodput maximization problem with quantized feedback. 
Assuming fixed $ \phi_i $, we cast the following CEP allocation problem
\begingroup
\allowdisplaybreaks
	\begin{align}
	\underset{r_i,\epsilon_i, \varepsilon_i}{\text{minimize}} ~~& \sum_{i=1}^{\Phi} r_i \int_{\phi_i}^{\phi_{i+1}} f(\gamma) \Omega(\gamma,r_i) \mathrm{d}\gamma
	\label{eq_opt_prob_6_obj}
	\\
	\text{s.t.} ~~ & \eqref{eq_opt_prob_5_const_1}, \eqref{eq_opt_prob_5_const_2}, \eqref{eq_opt_prob_5_const_3}, \eqref{eq_opt_prob_1_cons_1}  . \nonumber
	\end{align}
\endgroup
Notice that this problem is exponentially complex to find the optimal solution. However, we provide a suboptimal algorithm that can find a solution to the problem. 

Let us first assume that the optimum is achieved with equality for the constraint in Eq. \eqref{eq_opt_prob_5_const_1} noting that this assumption is a valid assumption for URLLC systems where $ \varepsilon_m $ is a small value, i.e. in the order of $ 10^{-6} $ or even lower. In this case, we relax the problem and introduce the following
\begingroup
	\begin{align}
	\underset{\varepsilon_i}{\text{minimize}} ~~ \sum_{i=1}^{\Phi} r_i \varepsilon_i ~~~~~ \text{s.t.} ~~  \eqref{eq_opt_prob_5_const_2}, \eqref{eq_opt_prob_5_const_3}, \eqref{eq_opt_prob_1_cons_1}  .
	\label{eq_opt_prob_7_obj}
	\end{align}
\endgroup
Assuming that $ r_i $ and $ \epsilon_i $ are independent variables, optimum can only be found if $ r_i $'s are all equal. In this case, the optimum is $ \varepsilon_i = \varepsilon_m/\Phi $. However, for any other cases no optimum solution can be found since the problem is unbounded and for any target value there would be feasible points with objective values smaller than the target.  

Recall that $ r_i $ and $ \epsilon_i $ are connected according to Eq. \eqref{eq_cep_ith_quant_reg}. Using this relationship a grid search algorithm can be proposed. However, a significant drawback here is the accuracy of the normal approximation since, as discussed in \cite{schiessl_thesis}, the normal approximation becomes inaccurate for very small blocklength, i.e., $ n< 100 $, and for low CEP , i.e., $ \epsilon < 10^{-5} $. Therefore, in this section, we propose a sub-optimal waterfilling-like algorithm where we set each $ \varepsilon_i $ inversely proportional to the truncated mean of $ \gamma $ within the quantized region. 

For the system model presented in Eq. \eqref{eq_system_model}, the optimum power selection for the $ m^{\text{th}} $ codeword transmission with the waterfilling algorithm can be found by solving \cite{kim_on_the, xu_on_the, makki_fast}
\begin{equation}\label{key}
P_m = \left[ \frac{1}{\delta} - \frac{1}{h_m^2} \right]^+ ,
\end{equation}
where $ [z]^+ = \max\{0,z\} $ and $ \delta $ satisfies
\begin{equation}\label{key}
\frac{1}{M} \sum_{m=1}^{M} \left[ \frac{1}{\delta} - \frac{1}{h_m^2} \right]^+ = nP ,
\end{equation}
where $ M $ represents the total number of transmissions. Similar to waterfilling algorithm, here, we introduce the scaling coefficient $ \alpha $ such that
\begin{equation}\label{key}
\alpha = \frac{\varepsilon_m}{\sum_{i=1}^{\Phi} \mathbb{E}\left[ \gamma_i | \phi_i \leq \gamma_i \leq \phi_{i+1} \right]^2}
\end{equation}
and the CEP constraint on the $ i^{\text{th}} $ quantization region is
\begin{equation}\label{eq_cep_assignment}
\varepsilon_i = \frac{\alpha}{\mathbb{E}\left[ \gamma_i | \phi_i \leq \gamma_i \leq \phi_{i+1} \right]^2} .
\end{equation}
Thus, the CEP constraint on the $ i^{\text{th}} $ quantization region is inversely related with the truncated mean of the random variable $ \gamma $. With this method, it is also guaranteed that $ \varepsilon_i > 0 $ for all $ i $ values. An illustration of the proposed method is depicted in Fig. \ref{cep_assignment} for Rician channel with $ K=10 $ when $ \Phi = 4 $. Notice that the quantization regions are separated at $ \phi_i = \{ 0, 0.5, 1, 2, \infty \} $ and  the truncated mean values are depicted with respect to the quantization regions. Based on the truncated means, the assigned CEP values can be seen in Fig. \ref{fig_cep_assignment_ricain_K10} with respect to the quantization regions, where $ \epsilon_m = 10^{-5} $. 

Next, we introduce a block coordinate descent based iterative algorithm where for each step the variables are updated respectively. 

\subsection{Problem 2: Optimization of Transmission Rate}

Given a set of $ \phi_i $'s, one can estimate the required CEP constraints on every $ i^{\text{th}} $ quantization region. Next, we search for the optimal $ r_i $ that maximizes the goodput while still maintaining the CEP constraints. The optimization problem can be written as
	\begin{align}
	\underset{r_i}{\text{maximize}} ~~ r_i \int_{\phi_i}^{\phi_{i+1}} f(\gamma) \big(1-\Omega(\gamma,r_i)\big) \mathrm{d}\gamma ~~~ \text{s.t.} ~~ \eqref{eq_opt_prob_5_const_1},  \eqref{eq_opt_prob_1_cons_1}, 
	\label{eq_opt_prob_8_obj}
	\end{align}
where the maximization operation is done over a single $ r_i $ since the overall maximization of the goodput is equivalent to the maximization of every $ i^{\text{th}} $ goodput separately. 

A two-step straightforward solution to this problem is: (\emph{i}) Follow Theorem 1 and find $ r_i' $ that maximizes the objective function without any constraint on CEP. (\emph{ii}) Check if $ r_i' $ satisfies the CEP constraint on the $ i^{\text{th}} $ quantization region. If it does not, gradually decrease $ r_i' $ until the  constraint is met. 

\begin{remark}
	If the CEP constraint cannot be achieved with an $ r_i > 0 $, one can only select $ r_i = 0 $ which yields the CEP of the $ i $th quantization region to be
	\begin{equation}\label{eq_cep_when_r_0}
	F(\phi_{i+1}) - F(\phi_{i}) .
	\end{equation}
	In this case, if $F(\phi_{i+1}) - F(\phi_{i}) > \varepsilon_i$,	no feasible solution can be achieved with this particular selections of $ \varepsilon_i $'s. A possible solution is to update $ \phi_i $'s and $ \varepsilon_i $'s accordingly.
\end{remark}

\subsection{ Problem 3: Optimization of Quantization Regions}

Next, we search for the optimal $ \phi_i $'s with the following optimization problem
\begin{subequations} 
	\begin{align} 
	\underset{\phi_i}{\text{maximize}}& ~~ \sum_{i=1}^{\Phi} r_i \int_{\phi_i}^{\phi_{i+1}} f(\gamma) \big(1-\Omega(\gamma,r_i)\big) \mathrm{d}\gamma  
    \\
	\text{s.t.}& ~~ \eqref{eq_opt_prob_5_const_1}, \eqref{eq_opt_prob_1_cons_2}.
	\end{align}
	\label{eq_opt_prob_9_obj}
\end{subequations}
This optimization problem belongs to the class of non-convex optimization problem, since the objective and the constraint in \eqref{eq_opt_prob_9_obj} are non-convex in $ \phi_i $. Although the Lagrangian duality theory has been proven to be an effective tool for solving the convex optimization problems with the primal or primal-dual computational methods, these methods are not powerful in a non-convex problem due to the non-zero duality gap between the primal and dual problems \cite{lavaei_zero}. In order to reduce the duality gap for non-convex problems the augmented Lagrangian method is introduced \cite{birgin_practical}. In this method, the duality gap is reduced by adding a penalized function, which consists of a quadratic term, to the Lagrangian.

The Lagrangian of the optimization problem in \eqref{eq_opt_prob_9_obj} can be defined by
\begin{multline}\label{key}
\mathcal{L} = \sum_{i=1}^\Phi r_i \left( F(\phi_{i+1}) - F(\phi_i) - \int_{\phi_i}^{\phi_{i+1}} f(\gamma) \Omega (\gamma , r_i ) \right) 
\\
- \sum_{i=1}^\Phi \lambda_i \left(  \int_{\phi_i}^{\phi_{i+1}}  f(\gamma) \Omega(\gamma,r_i) \mathrm{d}\gamma - \varepsilon_i  \right) 
\\
-\sum_{i=1}^\Phi \mu_i \left( \phi_i - \phi_{i+1} \right) ,
\end{multline}
for all $ \lambda_i \in \mathbb{R}^+ $ and $ \mu_i \in \mathbb{R}^+ $. The augmented Lagrangian method  eliminates the constraints and adds them into the objective function. Thus, the augmented Lagrangian function can be written as
\whencolumns{
	\begin{multline}\label{key}
	\mathcal{L}^a = \sum_{i=1}^\Phi r_i \left( F(\phi_{i+1}) - F(\phi_i) - \int_{\phi_i}^{\phi_{i+1}} f(\gamma) \Omega (\gamma , r_i ) \right) 
	\\
	-\frac{\rho}{2} \Bigg( \sum_{i=1}^\Phi \Bigg( \bigg[ \beta \left( \int_{\phi_i}^{\phi_{i+1}}  f(\gamma) \Omega(\gamma,r_i) \mathrm{d}\gamma - \varepsilon_i \right) 
	-\frac{\lambda_i}{\rho} \bigg]^+ \Bigg)^2 -\sum_{i=1}^\Phi \left( \left[ \phi_i - \phi_{i+1} - \frac{\mu_i}{\rho} \right]^+ \right)^2  \Bigg) ,
	\end{multline}
}{
	\begin{multline}\label{key}
	\mathcal{L}^a = \sum_{i=1}^\Phi r_i \left( F(\phi_{i+1}) - F(\phi_i) - \int_{\phi_i}^{\phi_{i+1}} f(\gamma) \Omega (\gamma , r_i ) \right) 
	\\
	-\frac{\rho}{2} \Bigg( \sum_{i=1}^\Phi \Bigg( \bigg[ \beta \left( \int_{\phi_i}^{\phi_{i+1}}  f(\gamma) \Omega(\gamma,r_i) \mathrm{d}\gamma - \varepsilon_i \right) 
	\\
	-\frac{\lambda_i}{\rho} \bigg]^+ \Bigg)^2 -\sum_{i=1}^\Phi \left( \left[ \phi_i - \phi_{i+1} - \frac{\mu_i}{\rho} \right]^+ \right)^2  \Bigg) ,
	\end{multline}
}
where $ \rho > 0 $ is the penalty parameter and $ \beta > 0 $ is the normalization factor since the the CEP values can be very small compared to the other penalty terms. $ \lambda_i $ and $ \mu_i $ are the Lagrangian dual variables associated with the constraints in  \eqref{eq_opt_prob_9_obj}. It is shown in \cite{bertsekas_nonlinear} that the augmented Lagrangian $ \mathcal{L}^a $ is locally convex when the penalty terms are sufficiently large. 

\begin{algorithm}[t]
    \small
	\caption{Goodput maximization with constraint on maximum CEP}
	\begin{algorithmic}[1]
		\State Initialize $ \phi_i $ for $ i=1,2,\cdots, \Phi+1 $  
		\State Set $ \varepsilon_i $ for $ i=1,2,\cdots, \Phi $ accordingly
		\Repeat 
		\For{$ i=1,2,\cdots, \Phi $} \label{marker}
		\State Update $ \varepsilon_i $ according to \eqref{eq_cep_assignment} 
		\EndFor
		\For{$ i=1,2,\cdots, \Phi $} 
		\State Solve \eqref{eq_theorem_1} for $ r_i $
		\If{$ \epsilon_i > \varepsilon_i $}
		\Repeat 
		\State $ r_i = r_i-\tau $
		\If {$ r_i == 0 $}
		\State Update $ \phi_i $ and $ \phi_{i+1} $, and \Goto{marker}
		\EndIf
		\Until {$ \epsilon_i \leq \varepsilon_i $}
		\EndIf
		\EndFor
		\For{$ i=\Phi,\Phi-1,\cdots, 2 $} 
		\State Set $ \lambda_i^{(1)}, \mu_i^{(1)} $, and $ \rho^{(1)} $
		\Repeat
		\State Solve \eqref{eq_augmented_lag_maximizer} for $ \phi_i $
		\State Update $ \lambda_i^{(l+1)}, \mu_i^{(l+1)} $, and $ \rho^{(l+1)} $ according to \eqref{eq_augmented_lag_update_lambda}, \eqref{eq_augmented_lag_update_mu}, and \eqref{eq_augmented_lag_update_rho}
		\Until \textit{Convergence}
		\EndFor
		\Until \textit{Convergence}
	\end{algorithmic} 
\end{algorithm}

The solution to the problem can now be approximated by solving
\begin{equation}\label{eq_augmented_lag_maximizer}
\underset{\phi_i}{\text{maximize}} ~~ \mathcal{L}^a .
\end{equation} 
  The essential feature of this method is that after each unconstrained maximization, which yields a maximization point $ \phi_i $, the variables $ \rho $,  $\lambda_i, $ and $ \mu_i $ are updated by means of the iteration. Thus, the Lagrangian dual variables are updated according to
\begin{equation}\label{eq_augmented_lag_update_lambda}
\lambda_i^{(l+1)} = \Bigg( \bigg[ \rho^{(l)} \beta \bigg( \int_{\phi_i}^{\phi_{i+1}}  f(\gamma) \Omega(\gamma,r_i) \mathrm{d}\gamma - \varepsilon_i \bigg) - \lambda_i^{(l)} \bigg]^+ \Bigg)^2,
\end{equation}
\begin{equation} \label{eq_augmented_lag_update_mu}
\mu_i^{(l+1)} = \left( \left[ \rho^{(l)}(\phi_i - \phi_{i+1}) - \mu_i^{(l)} \right]^+ \right)^2
\end{equation} 
and the penalty parameter $ \rho $ is updated as 
\begin{equation}\label{eq_augmented_lag_update_rho}
\rho^{(l+1)} = 2\rho^{(l)} ,
\end{equation}
where the upper-script in $ z^{(l)} $ represents the variable $ z $ at the $ l $th iteration step.\footnote{A discussion on the initialization of the Lagrangian dual variables can be found in \cite[Ch. 4]{birgin_practical}.} Notice that, with the augmented Lagrangian method, constraint violations are increasingly penalized as the penalty parameter $ \rho $ increases in every iteration step. This imposes the maximizer of Eq. \eqref{eq_augmented_lag_maximizer} to be in the feasible region. 

The proposed method is described in Algorithm 2. It can be seen that at each iteration, as $ \phi_i $'s and $ r_i $'s  change, $ \varepsilon_i $ values are also updated accordingly. The algorithm terminates as convergence is achieved. Otherwise, no feasible solution is found. Notably, the optimal solution is achieved by using some well-known optimization methods, which are the block coordinate descent type method and the augmented Lagrangian method, and the convergence of these methods are well documented in the literature, \cite{birgin_practical} and \cite{beck_on_the}.



\section{Numerical Results and Discussions}

\begin{figure*}[t]
	\centering
	\begin{subfigure}{.46\textwidth}
		\centering
		\includegraphics[width=1\linewidth]{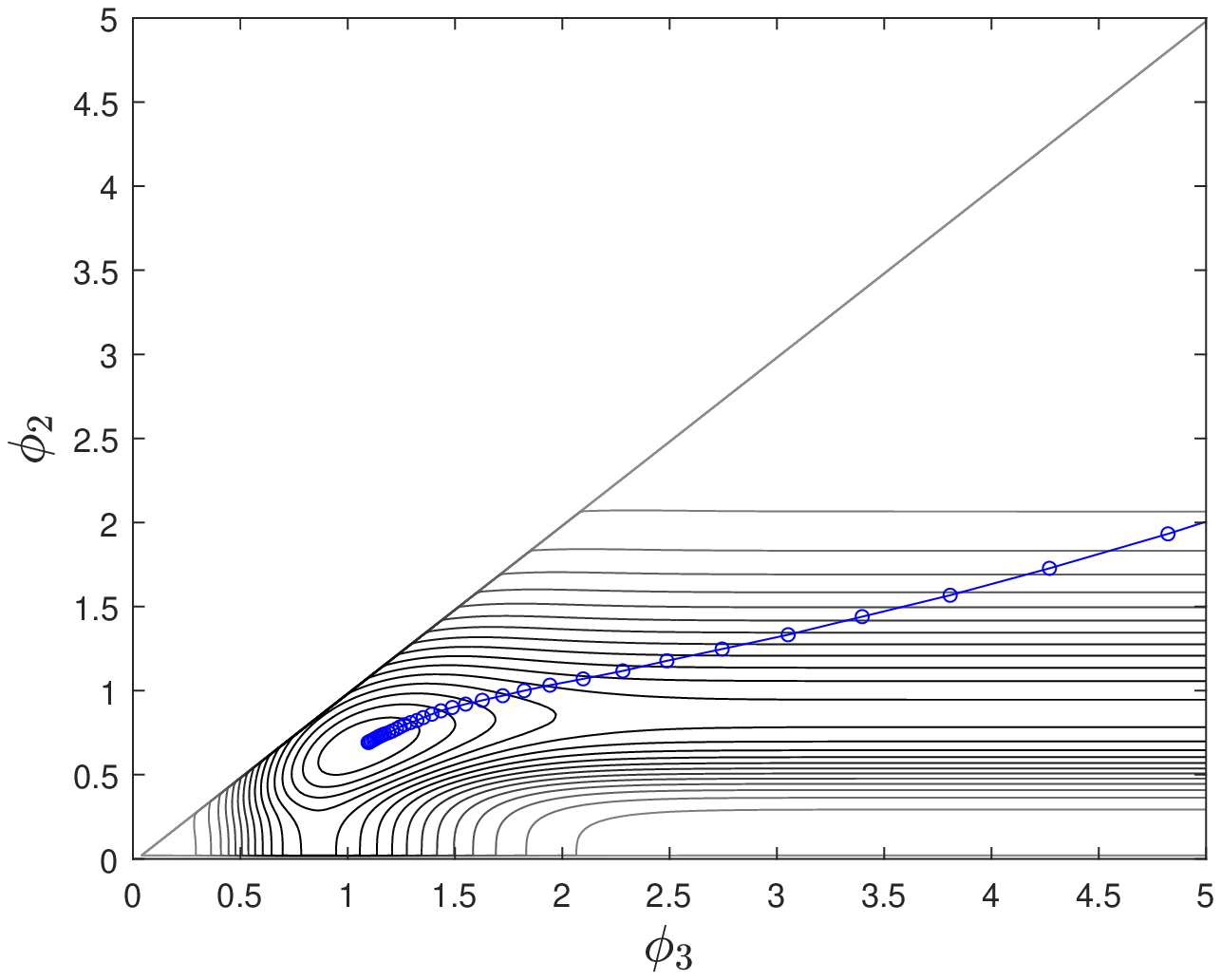}
		\caption{}
		\label{fig_contour_rician_K10_Q3}
	\end{subfigure}  ~~
	\begin{subfigure}{.5\textwidth}
		\centering
		\includegraphics[width=1\linewidth]{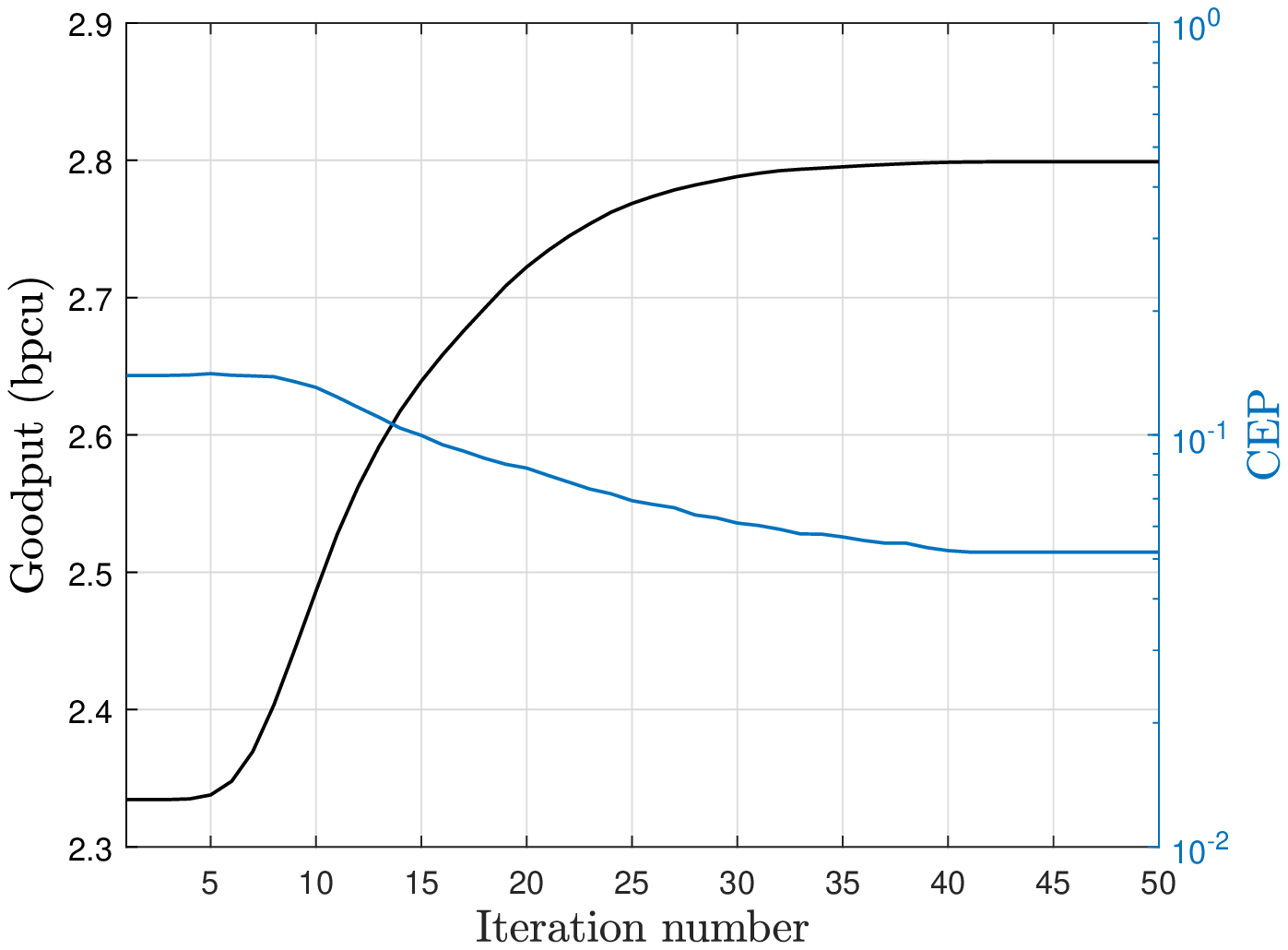}
		\caption{}
		\label{fig_iter_vs_exp_rate_and_cep}
	\end{subfigure}
	\caption{The behavior of the algorithm for the feedback scheme with $ \Phi = 3 $ and $ h $ is Rician distributed with $ K=10 $. (a) The contour of the goodput for all possible selections of $ \phi_2, \phi_3 \in [0, 5] $. The path of updated $ \phi_2 $ and $ \phi_3 $ in every subsequent iteration until convergence is shown with the blue line with circles. (b) (LHS) Achievable goodput and (RHS) the corresponding CEP of the system versus iteration number are depicted.}
	\label{fig_alg_1_behavior}
\end{figure*}


\begin{figure*}[t]
	\centering
	\begin{subfigure}{.5\textwidth}
		\centering
		\includegraphics[width=1\linewidth]{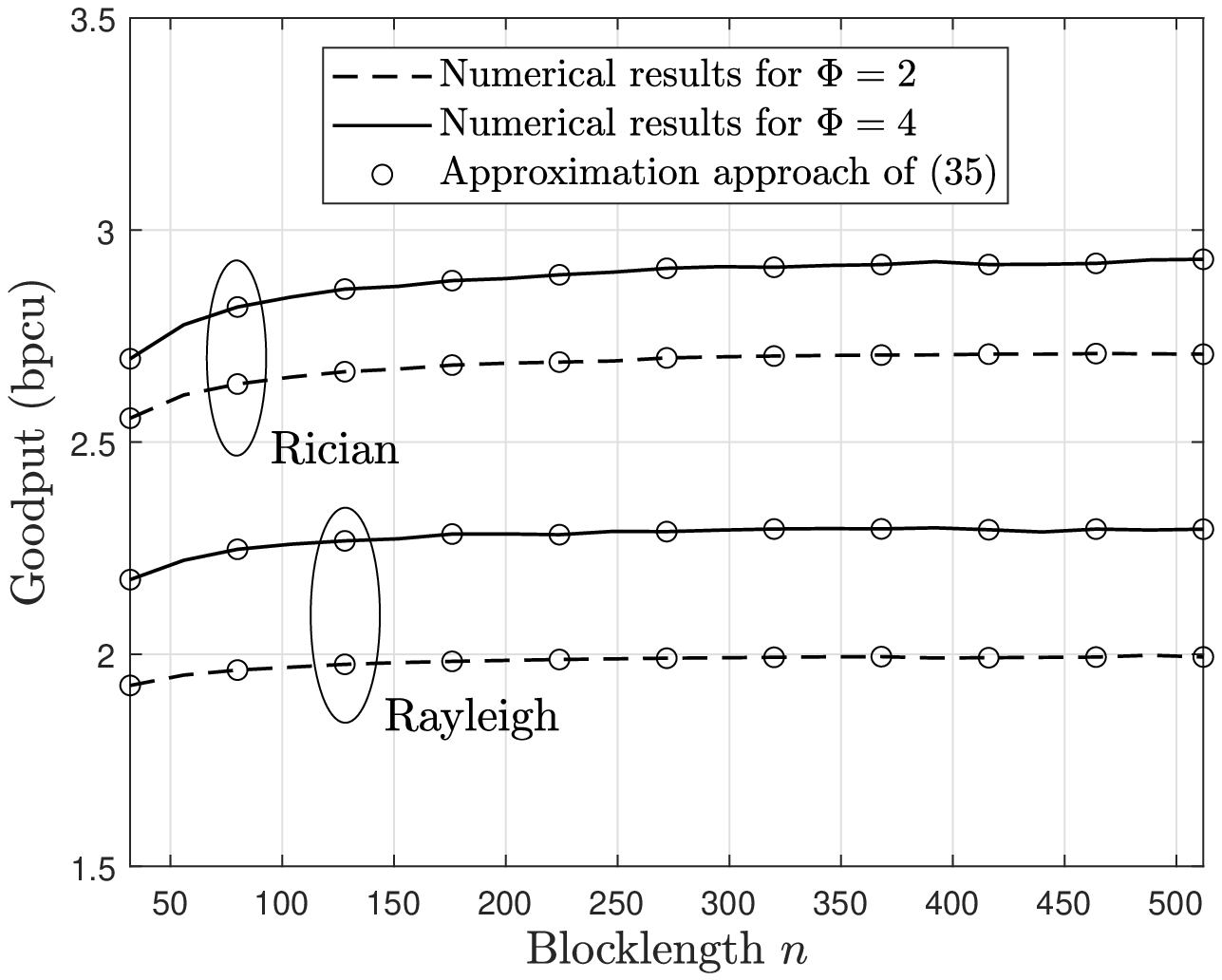}
		\caption{}
		\label{fig_comparison_of_approximation}
	\end{subfigure}%
	\begin{subfigure}{.5\textwidth}
		\centering
		\includegraphics[width=1\linewidth]{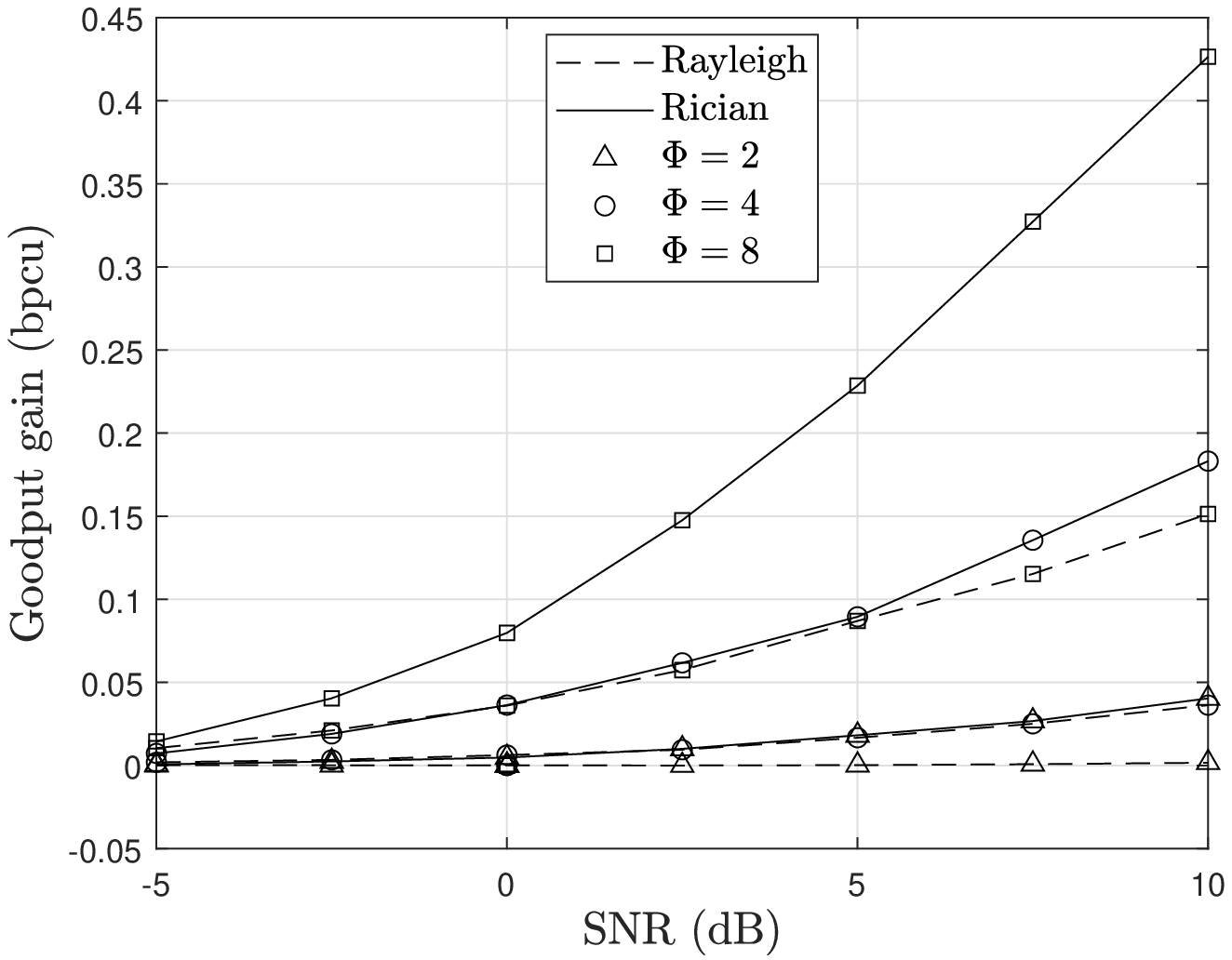}
		\caption{}
		\label{fig_benchmark_comparison}
	\end{subfigure}
	\caption[]{(a) Tightness of the proposed finite blocklength approximation in Theorem 2. Numerical comparisons for Rayleigh and Rician channels are depicted for $P = 10$ dB and $K = 10$. (b) Comparison of the proposed method with the benchmark. The goodput gain represents the difference between the achievable goodput values of the proposed method and the benchmark in bpcu. Results are computed for $n=64$, $P = 10$ dB and $K = 10$.}
\end{figure*}

The behavior of the Algorithm 1 for a simple case where $ \Phi = 3 $ is shown in Fig. \ref{fig_alg_1_behavior}, where the horizontal axes represent the quantization variables $ \phi_2 $ and $ \phi_3 $ and the vertical axis represents the goodput.\footnote{Notice that for $ \Phi = 3 $ only two quantization variables, which are $ \phi_2 $ and $ \phi_3 $, need to be optimized.} The contour surface in Fig. \ref{fig_contour_rician_K10_Q3} depicts the maximum achievable goodput according to the selections of $ \phi_2 $ and $ \phi_3 $. The solid line with circles shows the path of updated $ \phi_2 $ and $ \phi_3 $ in every subsequent iteration until the algorithm terminates.  Notice that no values are shown in the upper left half of Fig. \ref{fig_contour_rician_K10_Q3} since the numerical selections of $ \phi_2 $ and $ \phi_3 $ in this region violate the constraint in Eq. \eqref{eq_opt_prob_1_cons_2} and therefore no feasible result can be achieved. In this example, the initial values of $ \phi_2 $ and $ \phi_3 $ are selected to be $ \phi_2 = 5 $ and $ \phi_3 = 10 $. The channel coefficient $ h $ is selected to be distributed according to the Rician distribution with $ K=10 $. Selections of $ \phi_2 $ and $ \phi_3 $ in every iteration are shown with a circle in Fig. \ref{fig_contour_rician_K10_Q3}. As can be seen, the algorithm can efficiently track the path to the optimum and reach it.

\begin{figure*}[t]
	\centering
	\begin{subfigure}{.5\textwidth}
		\centering
		\includegraphics[width=1\linewidth]{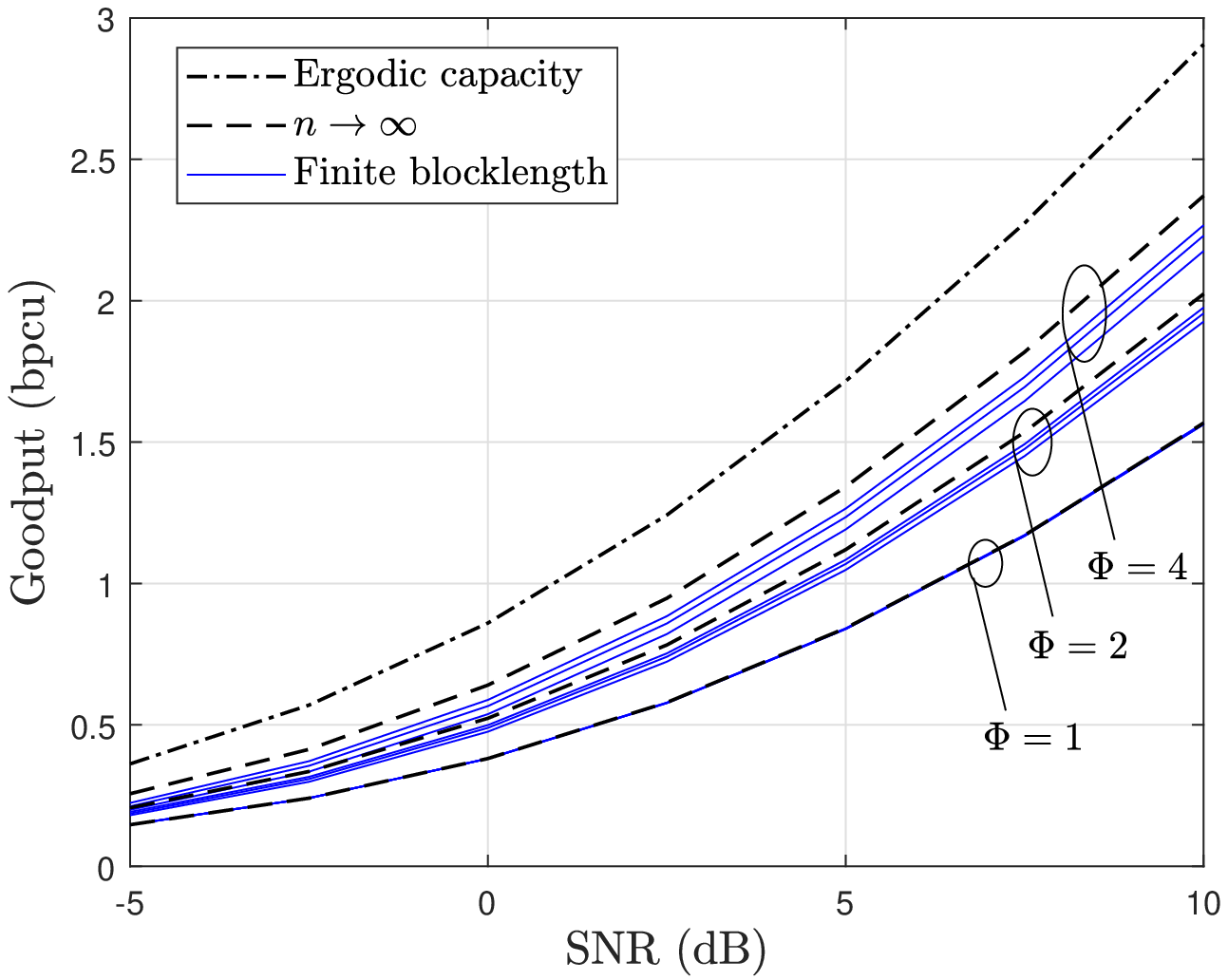}
		\caption{}
		\label{fig_gen_results_on_exp_throughput_rayleigh}
	\end{subfigure}%
	\begin{subfigure}{.5\textwidth}
		\centering
		\includegraphics[width=1\linewidth]{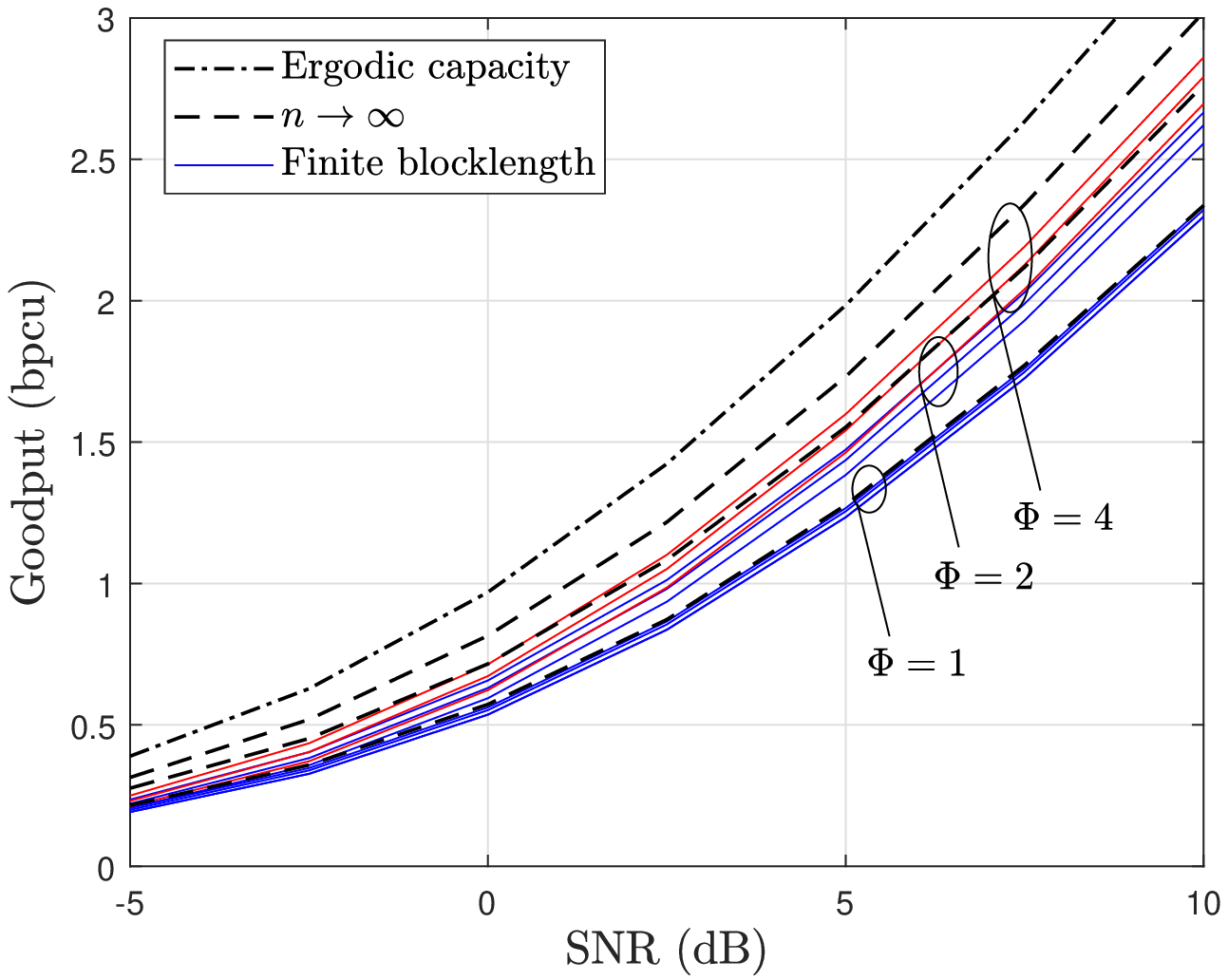}
		\caption{}
		\label{fig_gen_results_on_exp_throughput_rician_K_10}
	\end{subfigure}
	\caption[]{Maximum goodput achieved without any CEP constraint over Rayleigh and Rician channels with different feedback schemes, $ \Phi = \{1, 2, 4\} $. Values for finite $ n $ are depicted with blue solid lines, where we show results for $ n=\{32,64,128\} $. Notice that, for all cases, goodput values for $ n=128 $ are the highest ones and $ n=32 $ are the lowest. For comparison, the ergodic capacity and achievable goodput in asymptotic regime are also shown with dash-dotted and dashed lines, respectively. (a) Rayleigh channel. (b) Rician channel with $ K=10 $ dB.\footnotemark }
	\label{fig_max_exp_through_for_rayleigh_rician}
\end{figure*}
\footnotetext{ We used the red color to show the maximum achievable goodput values for $ \Phi = 4 $ in Fig. \ref{fig_gen_results_on_exp_throughput_rician_K_10} for finite blocklength results since they are mixed with the results of $ \Phi = 2 $. }

The corresponding goodput values after each iteration are also shown in Fig. \ref{fig_iter_vs_exp_rate_and_cep}, where one can see that the algorithm is able to find better selections of $ \phi_i $'s and $ r_i $'s in every iteration such that they lead to higher goodput and finally reach the optimal values. Moreover, we also show the overall CEP value of the system with the corresponding feedback scheme with the blue solid line in Fig. \ref{fig_iter_vs_exp_rate_and_cep}. It can be seen that as the goodput increases, the overall CEP decreases and reaches to a steady state after 50 iterations. 

The tightness of the approximation that is presented in Theorem 2 is shown in Fig. \ref{fig_comparison_of_approximation} for different Rayleigh and Rician channel conditions with different blocklength values. One can identify that Eq. (35) provides accurate approximation for the $Q$-function and the results of the algorithm are tight.  


It is also significant to compare the proposed method with proposed scheme with a benchmark scheme to show the superiority of the proposed method. Here, we derive the benchmark scheme by setting $\phi_i$'s according to the solution of the asymptotic regime that is presented in Eq. (11) and (12) but transmitting short blocklength at rate $R(n, \phi_i, \epsilon)$, from Eq. (13), when the channel gain is $\gamma \in [\phi_i, \phi_{i+1}]$. Goodput gain, which is the difference between the goodput value calculated with the proposed method and the benchmark, for Rayleigh and Rician channels are depicted in Fig. \ref{fig_benchmark_comparison} for $n=64$ and $\Phi=\{2,4,8\}$. It is important to observe that the gain is positive for all cases which means that the proposed method superior compared to the benchmark. Additionally, one can also see that the gain increases with higher SNR and higher quantization regions. And the gain is also higher in Rician channels compared to Rayleigh channel at the same SNR and $\Phi$, which is an expected result due to the effect of the channel dispersion.

Next, we show the maximum goodput values achieved by several feedback schemes with different numbers of quantization regions and blocklengths in Fig. \ref{fig_max_exp_through_for_rayleigh_rician}. Results for Rayleigh and Rician channels are depicted in Fig. \ref{fig_gen_results_on_exp_throughput_rayleigh} and Fig.  \ref{fig_gen_results_on_exp_throughput_rician_K_10}, respectively, for $ \Phi = \{1, 2, 4\} $ and $ n=\{32,64,128\} $. For comparison purposes, the ergodic capacity and achievable goodput values in the asymptotic regime \cite{kim_on_the} are also plotted. 

\begin{figure}[t]
	\centering
	\whencolumns{
		\includegraphics[width=.6\linewidth]{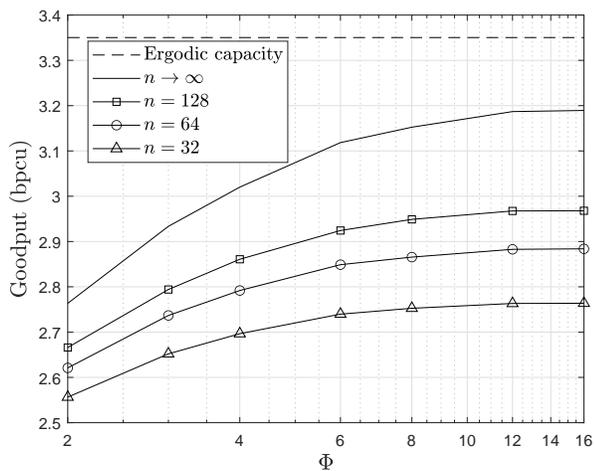}
	}{
		\includegraphics[width=1\linewidth]{goodput_w_capital_phi.eps}
	}
	\caption{The effect of $\Phi$ on the goodput. Results are computed for Rician channel with $K = 10$ and $P = 10$ dB.}
	\label{fig_goodput_w_capital_phi}
\end{figure}

\begin{figure*}[t]
	\centering
	\begin{subfigure}{.5\textwidth}
		\centering
		\includegraphics[width=1\linewidth]{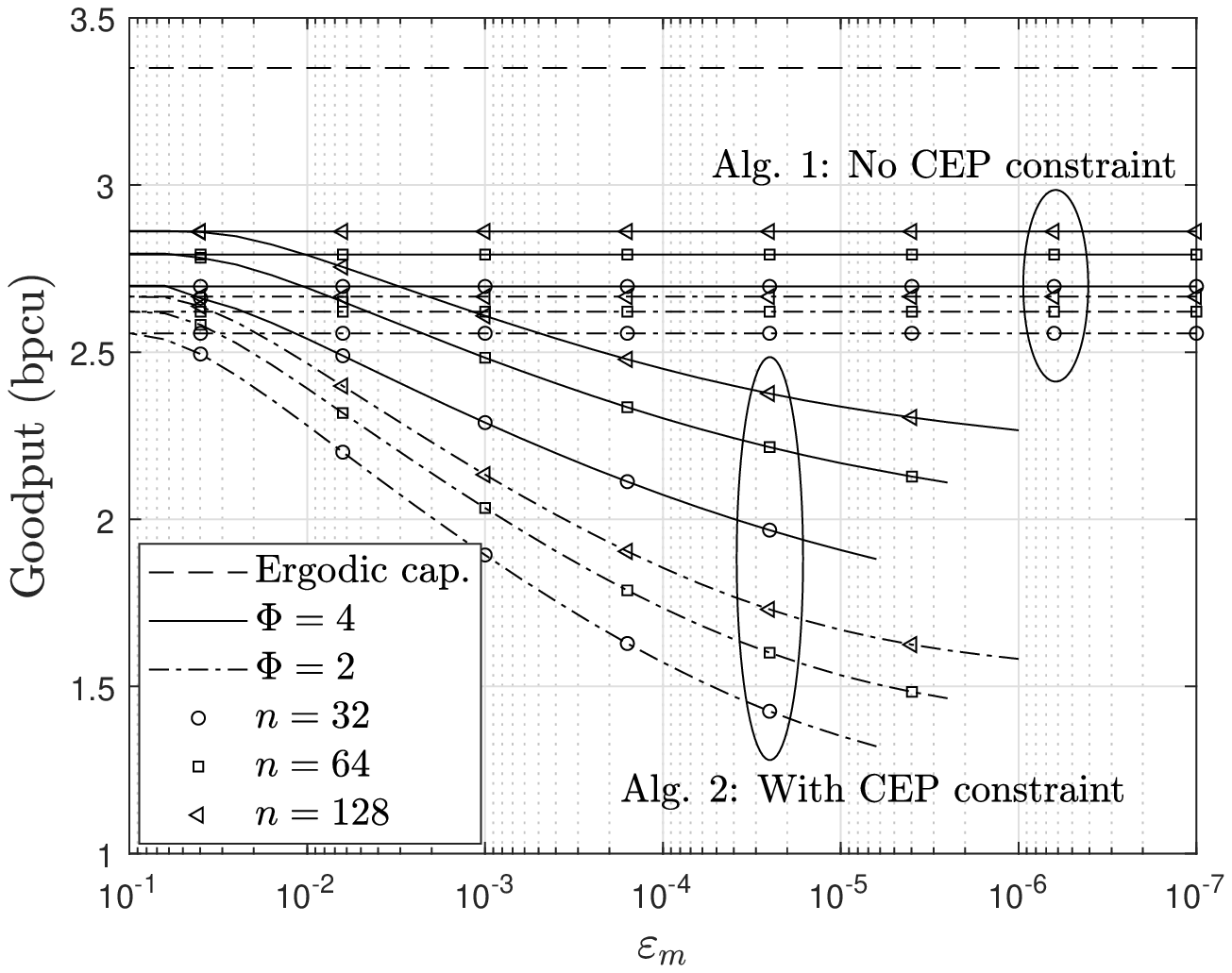}
		\caption{}
		\label{alg_1_vs_alg_2_rician_K_10_P_10}
	\end{subfigure}%
	\begin{subfigure}{.5\textwidth}
		\centering
		\includegraphics[width=1\linewidth]{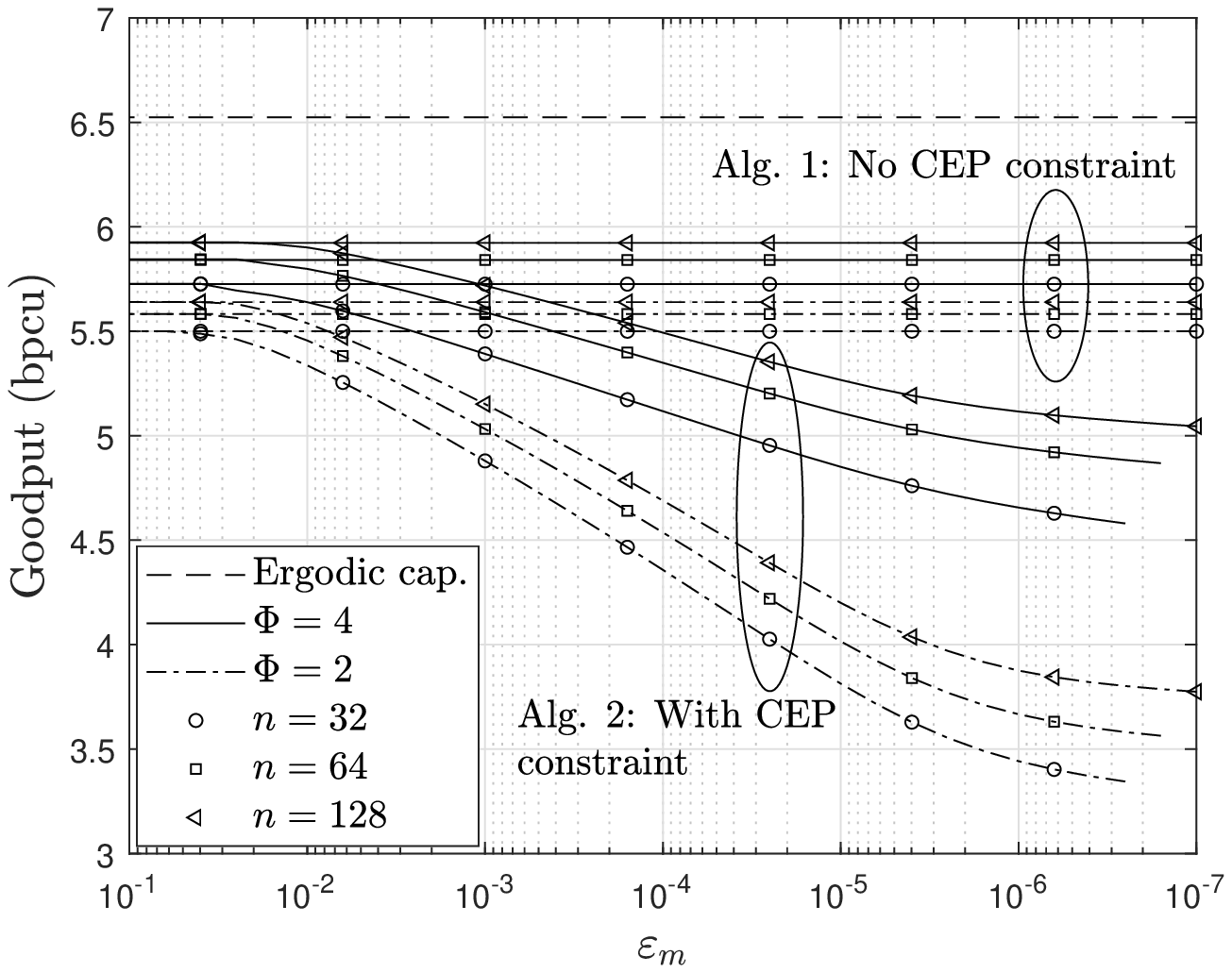}
		\caption{}
		\label{alg_1_vs_alg_2_rician_K_10_P_100}
	\end{subfigure}
	\caption{Maximum goodput values with and without constraint on CEP over Rician channel with $ K=10 $ dB and (a) $ P = 10 $ dB, (b) $ P=20 $ dB. Results for different feedback schemes with $ \Phi = \{2,4\} $ and $ n=\{32, 64, 128\} $ are depicted with respect to $ \varepsilon_m $, ranging from $ 10^{-1} $ to $ 10^{-7} $.}
	\label{alg_1_vs_alg_2_rician}
\end{figure*}

Notice that $ \Phi=1 $ represents the case where no CSI at the transmitter is available. Thus, it is important to observe that even with coarsely quantized systems, significant gains for both channel distributions are observed in both finite and infinite blocklength regimes. For example, in infinite blocklength regime, to achieve a target goodput of 1 bits per channel use (bpcu), feedback schemes with $ \Phi = 2 $ and $ \Phi = 4 $ require a power of roughly 2 and 3 dB less than a $ \Phi = 1 $ system does, respectively. Similar advantages are also observed for finite blocklength cases. It can be seen that with the quantized feedback scheme, the goodput of the communication system can be significantly increased. As expected, for all cases, i.e. fixed SNR and $ \Phi $, goodput with infinite blocklength are higher than finite blocklength. Thus, longer blocklength results in higher goodput. Except the case with Rayleigh distribution and $ \Phi = 1 $. In this case, differences among all schemes are negligible. This is due to the fact that the dispersion of the channel is zero, which is explained in detail in Sec. V. Therefore, the blocklength of the channel code does not have any impact on the achievable rate.  Furthermore, one can also observe that the gap between the finite and infinite blocklength cases is higher in Rician channel compared to Rayleigh. This is because channel dispersion is higher for Rician with high $ K $ and therefore blocklength has higher impact on the achievable rate. For instance, the goodput gap between infinite blocklength and $ n=128 $ for Rayleigh channel with $ \Phi=4 $ at $ 10 $ dB SNR is $ 0.11 $ bpcu. However, this gap increases to $ 0.16 $ bpcu for Rician channel.

An interesting and relevant investigation would be on the effect of the number of quantization regions, $\Phi$, on the goodput performance. This comparison is shown in Fig. \ref{fig_goodput_w_capital_phi}. Here, the achievable goodput values for Rician channel with $K=10$ and $P=10$ dB are depicted for various $n$ and feedback schemes with different $\Phi$ values. For comparison purposes the ergodic capacity, which does not depend on $\Phi$, and asymptotic results that are achieved with (11) and (12) are also added. The important outcome from Fig. \ref{fig_goodput_w_capital_phi} is that the goodput significantly increases as $\Phi$ increases when $\Phi$ is small, whereas this increment is not significant as $\Phi$ is big. Thus, based on the outcome of this scenario, one can say that the optimal $\Phi$ is 16, after which the performance does not significantly improve. 

\begin{figure}[t]
	\centering
	\whencolumns{
		\includegraphics[width=.6\linewidth]{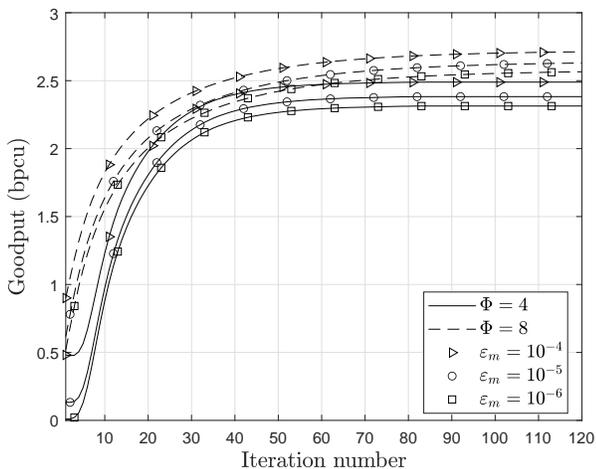}
	}{
		\includegraphics[width=1\linewidth]{convergence_alg_2.eps}
	}
	\caption{Convergence of Algorithm 2.}
	\label{alg_2_convergence}
\end{figure}

\begin{table}[b]
\centering
\begin{tabular}{|ccc|ccc|}
\hline
\multicolumn{3}{|c|}{$\Phi = 4$}  & \multicolumn{3}{c|}{$\Phi = 8$}  \\ \hline
\multicolumn{1}{|c|}{$\varepsilon = 10^{-4}$} & \multicolumn{1}{c|}{$\varepsilon = 10^{-5}$} & $\varepsilon = 10^{-6}$ & \multicolumn{1}{c|}{$\varepsilon = 10^{-4}$} & \multicolumn{1}{c|}{$\varepsilon = 10^{-5}$} & $\varepsilon = 10^{-6}$ \\ \hline
\multicolumn{1}{|c|}{0.0019}                  & \multicolumn{1}{c|}{0.0021}                  & 0.0103                  & \multicolumn{1}{c|}{0.0037}                  & \multicolumn{1}{c|}{0.0084}                  & 0.0213                  \\ \hline
\end{tabular}
\caption{Optimality gap between the proposed algorithm and the optimal solutions.}
	\label{table_opt_gap}
\end{table}

Up to now we showed the empirical results without any constraint on CEP. Next, we focus on the second optimization problem where restriction on CEP is imposed. In this case, we particularly focus on Rician channel with $ K =10 $ dB with $P = \{10, 20\} $ dB SNR. Maximum goodput values achieved by different schemes with CEP constraint are plotted in Fig. \ref{alg_1_vs_alg_2_rician} for $ n=\{32,64,128\} $ where the horizontal axis represents the maximum allowed CEP value $ \varepsilon_m \in [ 10^{-1}, 10^{-7} ]$. Algorithm 2 is used to solve the optimization problem with CEP constraint. Additionally, the ergodic capacity and achievable goodput values with no CEP constraint, which are attained by Algorithm 1, are also plotted for comparison purposes. Since these values do not change with $ \varepsilon_m $, they all have the same value for different $ \varepsilon_m $ and are represented with a straight horizontal line in Fig. \ref{alg_1_vs_alg_2_rician}.

Notice that, in Fig. \ref{alg_1_vs_alg_2_rician}, significant goodput gaps between the achievable values with and without CEP constraint are observed. This gap increases as $ \varepsilon_m $ decreases. For instance, to achieve $ \epsilon_m = 10^{-3} $ with $ n=128 $ when $ P=10 $ dB, it is required to introduce $ 0.53 $ and $ 0.26 $ bpcu of rate backoffs when $ \Phi = 2 $ and $ \Phi = 4 $, respectively. However, if $ \epsilon_m = 10^{-5} $, these gaps become $ 1 $ and $ 0.53 $ bpcu for $ \Phi = 2 $ and $ \Phi = 4 $, respectively. Thus, it is shown that imposing a restriction on CEP reduces the maximum goodput of the system. In other words, the feedback scheme achieves the required reliability by decreasing the transmission rate. Notably these gaps are even higher when $ P=20 $ dB. Furthermore, it can also be seen that for very low reliability values, i.e. $ \epsilon_m < 10^{-2} $, goodput results with CEP constraint converge to the results achieved without a CEP constraint, since the overall CEP that is achieved with Algorithm 1 is already lower than $ \varepsilon_m $. 

Recall that Algorithm 2 can only find a feedback scheme when a feasible set that fulfills the constraints is achieved, otherwise no feasible set can be constructed. This problem, for instance, arises for $ n=128 $ and $ P = 10 $ dB when $ \varepsilon_m < 10^{-6} $. In this case, no results are shown. Similarly, when $ n=64 $ and $ n=32 $, no solution is achieved when $ \varepsilon_m < 2.5\times 10^{-6} $ and $ \varepsilon_m < 6\times 10^{-6} $, respectively. One possible solution to this problem is to increase the SNR of the received codeword. Thus, as seen in Fig. \ref{alg_1_vs_alg_2_rician_K_10_P_100}, it is possible to achieve a feedback scheme for $ n=128 $ and $ \varepsilon_m = 10^{-7} $ when $ P = 20 $ dB. Similar results are observed for $ n=64 $ and $ n=32 $.

We finally discuss about the convergence of Algorithm 2. Goodput values achieved at every iteration are shown in Fig.  \ref{alg_2_convergence} for $\Phi = \{4, 8\}$, $n = 128$, and Rician fading channel with $K=10$, $P = 10$ dB under various reliability constraints. It is observed that the proposed algorithm converges and a steady state is achieved after 100 iterations in all cases. In addition, to see the optimality gap between the optimal solution and the scheme found by Algorithm 2, we performed an exhaustive grid search and list the gaps in Table \ref{table_opt_gap}.

\section{Conclusions}

Goodput maximization in the finite blocklength regime for a communication scenario with partial CSI knowledge at the transmitter over quasi-static channel is studied in this paper. An adaptive quantized feedback scheme that exploits the maximum goodput by searching the optimum selections of quantization regions with their corresponding transmission rates is formulated and investigated. This problem is then analytically solved and an iterative algorithm is proposed. Next, optimal quantized feedback scheme design for goodput maximization is studied under reliability constraint. A sub-optimum CEP allocation technique is presented. Then, the problem is analytically solved with the augmented Lagrangian method and an iterative algorithm that can achieve the maximum goodput while guaranteeing the reliability constraint is proposed. Lastly, numerical results of the optimization problem with and without constraint on CEP are presented. It is shown that significant improvement in goodput can be achieved even with coarsely quantized feedback schemes. However, differences between the achievable goodput results in asymptotic and non-asymptotic regimes are observed. It is shown that these differences vary with different channel distributions. Furthermore, we noticed that when constraint on reliability is imposed, the feedback scheme achieves the required reliability by decreasing the transmission rate and therefore maximum achievable goodput decreases with higher reliability constraint. Finally, it is worth mentioning that results that are presented in this paper are valid for any continuous channel distributions since the optimization problems are explicitly formulated without any specific channel distribution.

\bibliographystyle{IEEEtran}
\bibliography{FB_w_feedback}


\end{document}